\documentclass[english,11pt,a4paper]{amsart}
\usepackage{amstext,amsthm,amsfonts,amsmath,amssymb}
\usepackage{graphicx}
\usepackage{babel}
\usepackage{color} 
\usepackage{hyperref}
\usepackage{color}
\usepackage[latin1]{inputenc}

\makeatletter 

\newcommand{\N}{\mathbb{N}}

\newcommand{\R}{\mathbb{R}}
\newcommand{\C}{\mathbb{C}}
\newcommand{\Z}{\mathbb{Z}}

 \newcommand{\e}{\varepsilon}

\newtheorem{theorem}{Theorem}[section]
\newtheorem{corollary}[theorem]{Corollary}
\theoremstyle{definition}\newtheorem{definition}[theorem]{Definition}
\theoremstyle{definition}
\theoremstyle{definition}\newtheorem{remark}[theorem]{Remark}
\theoremstyle{definition}\newtheorem{example}[theorem]{Example}
\theoremstyle{definition}
\theoremstyle{definition}
\newtheorem{assumption}{Assumption}
\newtheorem{lemma}[theorem]{Lemma}

\title{Stability of the gapless pure point spectrum of self-adjoint operators}

\author{Paolo Facchi}
\address[Paolo Facchi]{Dipartimento di Fisica, Universit\`a degli Studi di Bari, I-70126 Bari, Italy, and INFN, Sezione di Bari, I-70126 Bari, Italy}
\author[Marilena Ligab\`o]{Marilena Ligab\`o$^*$}
\address[Marilena Ligab\`o]{Dipartimento di Matematica, Universit\`a degli Studi di Bari, I-70125 Bari, Italy}
\email{marilena.ligabo@uniba.it}
\thanks{$^*$ corresponding author}
\date{10/11/2022}

\begin{document}
\begin{abstract}
We consider a self-adjoint operator $T$ on a separable Hilbert space, with  pure-point and simple spectrum with accumulations at finite points. Explicit conditions are stated  on the eigenvalues of $T$ and on the bounded perturbation $V$ ensuring the global  stability  of the spectral nature of $T+V$.
\end{abstract}
\maketitle

\section{Introduction}
Let  $T$  be a self-adjoint operator on  a separable Hilbert space  ${\mathcal H}$ with domain $D(T)$.  Assume  that the spectrum of $T$ is pure point, i.e. $\sigma(T)=\sigma_{\textrm{pp}}(T)$, and simple.   Denote $(\tau_n)_{n\in\N}$ the  eigenvalues,   and $(e_n)_{n\in\N}\subset {\mathcal H}$ the corresponding orthonormal eigenvectors. 
 
This paper deals with the global stability of  $\sigma(T)$ under   bounded perturbations $V\in  {\mathcal L}({\mathcal H})$. We are interested in situations where the eigenvalues have accumulations at finite points, as in the extreme case of a dense pure point spectrum.   It will be shown, through  a quantum KAM iterative scheme  that, if $\e \in \R$ is small enough, $\sigma(T+ \e V)$ is globally stable, i.e. pure-point and simple,  for a large class of 
perturbations $V$, provided the following four assumptions are fulfilled:
 \begin{enumerate}
 \item[ (A.1)]  There exist $c \geq 1$ and $ \gamma >0 $ such that  for all $n, m \in \N, n\neq m$:
 \begin{equation}
 \label{A.1}
\frac{1}{|\tau_{m}-\tau_n|} \leq c \; |m-n|^{\gamma};
\end{equation}
 \item[(A.2)]  For all $n, m \in \N,  m > n$: 
 \begin{equation}
  \frac{1}{|\tau_{m}-\tau_n|} \leq  c \left( 1+ \frac{1}{\left|\tau_{m-n}-\tau_0 \right|}\right);  \label{A.1ter}
\end{equation}
\item[(A.3)]  For all $n, m \in \N,  m > n$, and for all $\ell \in \N^\ast=\N\setminus\{0\}$:
 \begin{equation}  \label{A.2}
 \left| \frac{1}{\tau_\ell-\tau_0}\left(\frac{1}{\tau_{m+\ell}-\tau_{n+\ell} } - \frac{1}{\tau_{m}-\tau_n} \right) \right| \leq  c \left( 1+ \frac{1}{\left|\tau_{m-n}-\tau_0 \right|}\right); 
\end{equation}
\item[ (A.4)] 
There exist $C, \alpha>0$ such that for all $k \in \N$:    
\begin{equation}
\label{A.3}
\sup_{m \in \N} |V_{m,m+k}|+\sup_{m \in \N, \,\ell \in \N^\ast} \left|\frac{V_{m+\ell,m+\ell+k}-V_{m,m+k}}{\tau_\ell -\tau_0}\right| \leq C e^{-\alpha k},
\end{equation}
where $V_{m,n}:=\langle  e_m,V e_n\rangle$, $m,n \in\N $.  
 \end{enumerate}

 \begin{remark} 
  The choice of  $\tau_0$ within $\sigma(T)$ in (A.2)--(A.4) is arbitrary. It can be replaced by  any other $\tau_p\in \sigma(T)$, up to the corresponding replacements:   $\tau_{m-n}-\tau_0\to \tau_{p+m-n}- \tau_p$ in formulas~\eqref{A.1ter} and~\eqref{A.2}, and $\tau_\ell-\tau_0\to \tau_{p+ \ell}- \tau_p$ in formula~\eqref{A.3}. 
  \end{remark}
 \begin{theorem}[Stability of pure point spectrum]
 \label{mainth}
Let $T(\varepsilon)=T+\varepsilon  V$ be the family of operators defined on $D(T)$, $\forall\,\varepsilon\in\R$.  Assume the validity of conditions (A.1)--(A.4). Then there exists $\e^*>0$ such that $\sigma(T(\e))$ is purely point and simple for all $\e\in\R$ with  $|\e| \leq \e^*$.
\end{theorem}

\begin{remark}
The {\it   local}  stability (in the sense of Kato~\cite{ref:Kato}) of the simple point spectrum follows from standard perturbation theory. Let $\tau$ be an isolated eigenvalue of $T$. The  geometrical condition  
\begin{equation}
\label{1.2}
\| \e V\| < \frac{d_\tau}{2} , \quad  d_\tau:=\mathrm{dist}(\tau, \; \sigma(T) \setminus \{\tau\}),
\end{equation}
entails the existence of  a simple isolated eigenvalue of $T+ \e V$  near $\tau$. 
The  {\it global} spectral stability is in turn  implied by  the {\it uniform} bound  
\begin{equation}
\label{1.3}
\| \e V\| < \frac{\delta}{2}, \quad \delta :=  \inf_{\tau\in\sigma(T)}d_\tau>0.
\end{equation}
Here, however, $\delta=0$. Condition (A.4)  takes care of of this  aspect. 
In the simplest situation,  namely with $V$  and $T$ admitting the same eigenvectors,  it reduces to
$$
\|V\| + \sup_{m \in \N, \,\ell \in \N^\ast} \left|\frac{V_{m+\ell,m+\ell}-V_{m,m}}{\tau_\ell -\tau_0}\right| \leq  C,
$$
which is a Lipschitz condition on $f$, with $f(\tau_m):=V_{m,m}$. 
In particular, if  a subsequence   $(\tau_{\ell_j})_{j \in \N}\subset \sigma(T)$ has a limit point for $j \to +\infty$,  the same must be true for the  corresponding sequence  $(V_{m+\ell_j})_{j \in \N}$ for any fixed $m$. In other words,  the nature of $\sigma(T)$ is preserved if  $V$ is \emph{shaped} like~$T$.

\end{remark}

 This concept  can be  extended to the  general case:
 
\begin{definition}
\label{def:shapeL}
Let $V\in  {\mathcal L}({\mathcal H})$ and $V_{m,k}=\langle  e_m,V e_k\rangle$ be its matrix elements in the eigenbasis of $T$. 
We say that $V$ is $T$-shaped 
if for all $k \in \N$:
$$
\sup_{n \in \N, \,\ell \in \N^\ast} \left|\frac{V_{n+\ell,n+\ell+k}-V_{n,n+k}}{\tau_\ell -\tau_0}\right| <+\infty
$$
and $(C,\alpha)$-exponentially decreasing $T$-shaped if in addition (A.4) above holds. 
\end{definition}

In this way Theorem~\ref{mainth}  
 can be reformulated as follows:
 \par\noindent 
 \textit{The operator  $T(\e)=T+  \e V$, $\e \in \R$, has a simple pure-point spectrum provided $\sigma(T)$ fulfills  conditions (A.1)--(A.3),
  $V$ is  $(C,\alpha)$-exponentially decreasing $T$-shaped, and  $\e$ is small enough. } 
 \par\noindent
 
The $T$-shaping of the perturbation $V$ allows to apply a variant of the KAM scheme~\cite{ref:KAM}: the small denominators generated by the procedure are controlled through a propagation of the  Diophantine conditions (A.1)  made possible by the appropriate shaping (A.4) of the perturbation according to  $\sigma(T)$. 

The Article is organized as follows: In Sec.~\ref{sect:examples} we state Theorem~\ref{thm:mainth2}, that is a more general formulation of Theorem~\ref{mainth}, and we discuss the advances of our result with respect to the literature, while in Sec.~\ref{sectionex} we exhibit some concrete examples. 
As a preliminary step, in  Sect.~\ref{stabilityden}  we prove that Assumption~(A.4)  implies  stability, under the perturbation,  of the Diophantine conditions~(A.1) on the eigenvalues $(\tau_n)_{n\in \N}$. In Sec.~\ref{sect:proof} we prove Theorem~\ref{thm:mainth2} using  a quantum KAM scheme. Finally, in Appendix~\ref{sect:lemmas} we prove some supplementary technical Lemmas.

\section{Stability of the pure point spectrum of operators on $\ell^2(\Z^d)$}  \label{sect:examples}
Theorem~\ref{mainth} is a  straightforward corollary of the following theorem on the stability of the pure point spectrum of operators on $\ell^2(\Z^d)$, $d\geq 1$. Let  $T=T^\ast$ be a multiplication operator  on $\ell^2(\Z^d)$ with purely point spectrum and simple  eigenvalues,  denoted by $(\lambda_n)_{n\in\Z^d} \subset \R$, corresponding to the canonical basis $(e_n )_{n \in \Z^d}$ of $\ell^2(\Z^d)$, with $e_n=(\delta_{m,n})_{m \in \Z^d}$, so that
\begin{equation}
 D(T)=\left\{v=(v_n)_{n \in \Z^d}\in \ell^2(\Z^d)\,:\,\sum_{n\in\Z^d} \lambda_n^2 |v_n|^2<+\infty\right\}, \qquad T v = \sum_{n\in\Z^d} \lambda_n v_n,	
\end{equation}
and let  $V \in {\mathcal L\left(\ell^2(\Z^d)\right)}$.   Assume that: 
\begin{assumption}\label{A.1new}
There exist $c \geq 1$ and $ \gamma >0 $ such that  for all $k \in \Z^d_0:= \Z^d \setminus \{0\}$:
 \begin{equation}\label{A.1d}
\sup_{n \in \Z^d}  \frac{1}{|\lambda_{n+k}-\lambda_n|} \leq c \; |k|^{\gamma},
\end{equation}
where $|k|:=|k_1|+\dots+|k_d|$;
\end{assumption}
\begin{assumption}\label{A.2new}
For all $k \in \Z^d_0$:
\begin{equation}\label{A.2.1d}
\sup_{n \in \Z^d} \frac{1}{|\lambda_{n+k}-\lambda_n|} \leq
  c \left( 1+ \frac{1}{|\lambda_k-\lambda_0|}\right);  
\end{equation}
\end{assumption}
\begin{assumption} \label{A.3new}
For all $k \in \Z^d_0$:
 \begin{equation}\label{A.2.2d}
 \sup_{ n\in \Z^d, \, j \in \Z^d_0} \left| \frac{1}{\lambda_j-\lambda_0}\left(\frac{1}{\lambda_{n+j+k}-\lambda_{n+j} } - \frac{1}{\lambda_{n+k}-\lambda_n} \right) \right| \leq
 c \left( 1+ \frac{1}{|\lambda_k-\lambda_0|}\right);  
\end{equation}
\end{assumption}
\begin{assumption}\label{A.4new}
There is an $\alpha>0$ such that
\begin{equation}
\label{A.3d}
\sup_{k\in\Z^d}e^{\alpha |k|}\left(\sup_{m \in \Z^d} |V_{m,m+k}|+\sup_{m \in \Z^d, \,j \in \Z^d_0} \left|\frac{V_{m+j,m+j+k}-V_{m,m+k}}{\lambda_j -\lambda_0}\right|\right)  < +\infty ,
\end{equation}
where $V_{m,n}=\langle  e_m,V e_n\rangle$, $m,n \in\Z^d $ are  the matrix elements of $V$ in the  canonical basis.  
\end{assumption}
Our main result is the following.
 \begin{theorem}[Stability of pure point spectrum for operators on $\ell^2(\Z^d)$]
 \label{thm:mainth2}
Let $T(\varepsilon)=T+ \e V$ be the family of  operators defined on $D(T)$, $\forall\,\varepsilon\in\R$.  If Assumptions~\ref{A.1new}--\ref{A.4new}  hold, then there exists $\e^*>0$ such that  if $|\e| \leq \e^*$  then:
\begin{itemize}
\item The spectrum of $T(\e)$ is pure point and simple and its  eigenvalues $(\lambda_n(\e))_{ n\in\Z^d}$ fulfill the Diophantine condition
\begin{equation}
\label{Diofanto3}
\frac{1}{|\lambda_{m}(\e)-\lambda_n(\e)|} \leq 12c^2 \, |m-n|^{2\gamma}, \qquad \forall m,n \in \Z^d, m\neq n;
\end{equation}
\item There exists an invertible $U \in \mathcal{L}\left(\ell^2(\Z^d)\right)$  such that
\begin{equation}
U^{-1}(T+\e V)U=T^{(\infty)},
\end{equation}
where $T^{(\infty)} e_n=\lambda_n(\e) e_n$, for all $n \in \Z^d$;
\item Moreover, if $V^*=V$, then $U$ is unitary and $T(\e)$ has a complete set of exponentially localized orthonormal eigenvectors $(u_n(\e))_{n \in \Z^d}$ with localization length $(\alpha-\sigma)^{-1}$, where $\sigma= \min \{1,\alpha/2\}$, i.e.  for all $n \in \Z^d$ there is $C_n >0$ such that:
\begin{equation}\label{eqn:exploc}
|\langle e_j, u_n(\e) \rangle| \leq C_n e^{-(\alpha-\sigma)|j-n|}, \qquad \textrm{for all $j \in \Z^d$}.
\end{equation}
\end{itemize}
\end{theorem}
\begin{remark}[Explicit bounds]\label{rem:xieps}
The proof of Theorem~\ref{thm:mainth2}  is iterative and provides a value for $\e^*$, see~(\ref{eqn:hypeps}):
\begin{equation}\label{epstar}
\e^* = \frac{A \min\left\{ 1, \frac{\alpha}{2}\right\}^{4d+2\gamma}}{\|V\|_{\alpha}},
\end{equation}
where 
\begin{equation}
A = \frac{1}{ 12c^2 \left(\frac{2\gamma}{e }\right)^{2\gamma} 4^{4d+2\gamma}}\min\left\{ \frac{2^{4(2d+\gamma)} 3 c \left(\frac{2\gamma}{e }\right)^{2\gamma}  }{1+2^{4(2d+\gamma)} 3 c \left(\frac{2\gamma}{e }\right)^{2\gamma} }, 1-\frac{3^d}{2^{6d-1}}\right\} \label{Acost}
\end{equation}
and 
\begin{equation}
	\|V\|_{\alpha}= \sup_{k\in\Z^d}e^{\alpha |k|}\left(\sup_{m \in \Z^d} |V_{m,m+k}|+\sup_{m \in \Z^d, \,j \in \Z^d_0} \left|\frac{V_{m+j,m+j+k}-V_{m,m+k}}{\lambda_j -\lambda_0}\right|\right) \label{eqn:xi}
\end{equation}
is the left hand side of Assumption~\ref{A.4new}.
Moreover, at the end of the proof of Theorem~\ref{thm:mainth2}  we obtain also an explicit value for the constant $C_n$ in~(\ref{eqn:exploc}), see~(\ref{eqn:Cn}). Notice that the value of $\e^*$ in~(\ref{epstar}) could be optimized by using the technique of approximation functions in~\cite{ref:Russmann}, see the next Remark.
\end{remark}
\begin{remark}[Approximation functions]
Assumption~\ref{A.1new} can be replaced with the following, more general one: there exist $c \geq 1$ and a continuous function $\Omega:[0,+\infty) \to [1,+\infty)$  such that
 \begin{equation}\label{A.1Omega}
\sup_{n \in \Z^d}  \frac{1}{|\lambda_{n+k}-\lambda_n|} \leq c \; \Omega(|k|).
\end{equation}
The function $\Omega$, called ``approximation function''~\cite{ref:Russmann}, must satisfy the following property: both these functions
$$
\Phi(x)= x^{-4d} \sup_{r \geq 0}\Omega(r)e^{-x r}
$$
and 
$$
\Psi(x) = \inf_{\mathcal{S}_{x}} \prod_{j=0}^{+\infty}\Phi(x_j)^{\frac{1}{2^{j+1}}},
$$
are finite for $x>0$, where $\mathcal{S}_{x}$ is the set of all non-negative decreasing sequences $(x_j)_{j\in \N}$ such that $\sum_{j}x_j \leq x$. The case of Assumption~\ref{A.1new} corresponds to $\Omega(r)=|r|^\gamma$. If we replace Assumption~\ref{A.1new} with condition~(\ref{A.1Omega}), Theorem~\ref{thm:mainth2} is still valid, with the replacement of~(\ref{Diofanto3}) with
\begin{equation}
\frac{1}{|\lambda_{m}(\e)-\lambda_n(\e)|} \leq 12c^2 \, \Omega(|m-n|)^{2}, \qquad \forall m,n \in \Z^d, m\neq n.
\end{equation}
\end{remark}

An interesting application of Theorem~\ref{thm:mainth2} is to the family of  discrete Schr\"{o}dinger operators.
 \begin{corollary}[Discrete Schr\"odinger operators]
 \label{thm:corollary}
Let $\Delta$ be the discrete Laplacian on $\ell^2(\Z^d)$, i.e. $\langle e_m, \Delta e_n\rangle= \delta_{|m-n|,1}$,  for all $m,n \in \Z^d$, and let $\e \in \R$. If Assumptions~\ref{A.1new}--\ref{A.3new}  hold and
\begin{equation}\label{hypeps}
 |\e|\leq  Ae^{-2},
\end{equation}
with $A$  as in~(\ref{Acost}), then  the spectrum of the discrete Schr\"{o}dinger operator $T+ \e \Delta$ is pure point and simple and its  eigenvalues $(\lambda_n(\e))_{ n\in\Z^d}$ fulfill the Diophantine condition
\begin{equation}
\label{Diofanto3D}
\frac{1}{|\lambda_{m}(\e)-\lambda_n(\e)|} \leq 12c^2 \, |m-n|^{2\gamma}, \qquad \forall m,n \in \Z^d, m\neq n.
\end{equation}
Moreover, $T+ \e \Delta$ has a complete set of exponentially localized orthonormal eigenvectors $(u_n(\e))_{n \in \Z^d}$, i.e.  for all $n \in \Z^d$ there is $C_n >0$ such that:
\begin{equation}\label{eqn:explocD}
|\langle e_j, u_n(\e) \rangle| \leq C_n e^{-(\alpha-1)|j-n|}, \qquad \textrm{for all $j \in \Z^d$},
\end{equation}
where $\alpha= \log (A/|\e|)\geq 2$. 
\end{corollary}
\begin{proof}
Notice that $\Delta$ satisfies Assumption~\ref{A.4new} for all $\alpha >0$, since an easy computation gives $\|\Delta\|_{\alpha}=e^{\alpha}$, where $\|\Delta\|_{\alpha}$ is defined in~(\ref{eqn:xi}). Moreover, Eq.~\eqref{epstar} reads
$$
\e^* = A e^{-\alpha} \min\left\{ 1, \frac{\alpha}{2}\right\}^{4d+2\gamma},
$$
whose maximum value, obtained at $\alpha=2$, is $\e^*=A e^{-2}$, that is the right hand side of~(\ref{hypeps}). Then the proof follows immediately by Theorem~\ref{thm:mainth2}, by choosing the maximum possible value of $\alpha$.
\end{proof}

\begin{remark}[Quasi-periodic Schr\"odinger operators]
	
There are several results on the perturbation of pure point spectrum for the discrete Schr\"{o}dinger operators with quasiperiodic potentials, see e.g.~\cite{ref:Jitomirskayabook, ref:Goldstein} for a review and the reference therein. These operators acting on $\ell^2(\Z)$ are defined as follows:  
\begin{equation}\label{eqn:Mathieu}
H(\omega, \lambda, \theta) e_n= \lambda g (\omega n+ \theta)e_n+e_{n-1}+e_{n+1}, \quad n \in \Z,
\end{equation}
where $\lambda, \omega, \theta \in \R$  and $g$ is an analytic $1$-periodic function.
The case $g(x)=\cos (2 \pi x)$, corresponding to the almost Mathieu operator (also known as Harper's operator), has been extensively studied~\cite{ref:Dinaburg, ref:Lima, ref:Sinai, ref:Spencer, ref:Eliasson, ref:Jitomirska}. In particular, in~\cite{ref:Lima} the authors proved the existence of some pure point spectrum for $\lambda$ very large (and of some absolutely continuous spectrum for $\lambda$ very small) for a.e. $\omega$. In~\cite{ref:Sinai, ref:Spencer}  there are the first proofs of complete localization (pure point spectrum with exponentially decaying eigenfunctions) for $\lambda$ large (a.e. $\omega, \theta$). Later, alternative KAM-type arguments for $\lambda$ large were developed in~\cite{ref:Eliasson}.
In~\cite{ref:Jitomirska, ref:Gordon}, the authors  prove, in a non-perturbative way, that for almost every $\omega, \theta \in \R$ the almost Mathieu operator $H(\omega, \lambda, \theta)$ has:
\begin{itemize}
\item For $\lambda >2$, only pure point spectrum with exponentially decaying eigenvectors;
\item For $\lambda=2$, purely singular-continuous spectrum;
\item For $\lambda<2$, purely absolutely continuous spectrum.
\end{itemize}
In~\cite{ref:Bourgain} the authors prove that for any analytic $1$-periodic function $g$ and any Diophantine number $\omega \in \R$, i.e. there are $b >0$ and $\zeta > 1$ such that for all $j \in \Z$ and $m\in \Z$:
$$
|\omega \cdot m+j| \geq \frac{b}{|m|^\zeta}, 
$$
and a.e. $\omega \in \R$, there exists  a constant $\lambda(g)>0$ such that the operator in~(\ref{eqn:Mathieu}) has purely absolutely continuous spectrum for $\lambda > \lambda(g)$.

A different approach is used in~\cite{ref:Poeschel, ref:Bellissard} where the authors mainly consider the operator $H(\lambda)$ on $\ell^2(\Z^d)$: 
\begin{equation}\label{eqn:discreteS}
H(\lambda)e_n=\lambda f(\omega \cdot n) e_n + \sum_{|k|=1} u_{n+k}, \quad n \in \Z^d,
\end{equation}
where $f$ is a $1$-periodic meromorphic  function on $\{z \in \C: |\mathrm{Im}z| < R \}$, for some $R>0$ , and $\omega \in \R^d$ a Diophantine vector, i.e. there are $b >0$ and $\zeta > d$ such that for all $j \in \Z$ and $m=(m_1, \dots, m_d) \in \Z^d$ such that
$$
|\omega \cdot m+j| \geq \frac{b}{|m|^\zeta}.
$$
 Under suitable assumptions on $f$, they prove, using a KAM scheme, that there exists $\lambda(f)>0$ such that for all $\lambda > \lambda(f)$ the operator in ~(\ref{eqn:discreteS}) has purely point spectrum with exponentially localized eigenvectors. More precisely, the result in~\cite{ref:Poeschel} is stated at an abstract level and it is concretely applied only to the case $f(x)= \tan(\pi x)$.
In~\cite{ref:Bellissard} the result is actually stated and proved for a more general class of operators
\begin{equation}\label{eqn:discreteSg}
H(\lambda)e_n=\lambda f(\omega \cdot n) e_n + V e_n, \quad n \in \Z^d,
\end{equation} 
with $V \in \mathcal{L}(\ell^2(\Z^d))$ satisfying suitable assumptions. 

\end{remark}

Notice that in all the results stated above the eigenvalues of the unperturbed operators have the same algebraic structure: they are always obtained by sampling periodic functions (with some regularity in the complex plane) on points $x_{n}=\omega \cdot n+ \theta$, $n \in \Z^d, \omega \in \R^d, \theta \in \R$.

The operator family $T(\e)=T+ \e V$ considered in Theorem~\ref{thm:mainth2} includes the operators in~(\ref{eqn:Mathieu}) and~(\ref{eqn:discreteSg}), but it is more general because the eigenvalues of $T$ do not obey to the aforementioned algebraic structure. Moreover, the class of possible perturbations $V$ is explicitly expressed in terms of the eigenvalues of $T$ and does not require any regularity in the complex plane, see Sec.~\ref{sectionex}.

\section{Examples}\label{sectionex}
In this Section we consider  some examples of eigenvalues $(\lambda_n)_{n \in \Z^d}$ satisfying Assumptions~\ref{A.1new}--\ref{A.3new}, and then  several examples of perturbations $V$ satisfying Assumption~\ref{A.4new}.  
\begin{example}\label{ex:dirich}
Let $d \in \N$, $d \geq 2$,   $\omega:=(\omega_1, \dots,\omega_d) \in \R^d$ such that $\omega \cdot n \neq 0$ for all $ n \in \Z^d_0$. 
The function
\begin{equation}
\label{eq:mundef}
n\in\Z^d \mapsto \mu_n := \omega\cdot n \in \R
\end{equation}
is  injective, with $\mu_0=0$.  Assume that there exist $C>0$ and  $\gamma >0$ such  that 
\begin{equation}
 \label{eq:mundiof}
\frac{1}{|\mu_k|} = \frac{1}{|\omega\cdot k|}\leq  C |k|^\gamma, \quad  \forall\,k \in \Z^d_0 .
\end{equation}
Since $\mu_{n+k} = \mu_n + \mu_k$ for all $n, k \in \Z^d$, we have that 
\begin{eqnarray*}
&&
\frac{1}{|\mu_{n+k}-\mu_n|} = \frac{1}{|\mu_k|}\leq  C |k|^\gamma,
\\
&&
\frac{1}{\mu_j-\mu_0}\left(\frac{1}{\mu_{n+j+k}-\mu_{n+j} } - \frac{1}{\mu_{n+k}-\mu_n} \right) = 0,
\end{eqnarray*}
for all $n, k, j \in \Z^d$, $k,j  \neq 0$.
Hence  Assumptions~\ref{A.1new}--\ref{A.3new} are satisfied with $c=\max\{C,1\}$.
\end{example}

This example generates several others by the following procedure: let $h:\R\to\R$ and denote with  $S_\omega :=\{ \omega \cdot n\, :\, n \in \Z^d\}$ the range of the function in~(\ref{eq:mundef}). Define 
\begin{equation}
\label{lambdan}
\lambda_n := h(\mu_n), \qquad n\in\Z^d,
\end{equation}
and {\it assume} $h|_{S_\omega}$  injective, so that  $\lambda_n\neq \lambda_m$ if $n\neq m$.
In particular, injectivity is ensured by either one of the following conditions:
\begin{itemize}
\item $h:\R\to\R$ is injective;
\item $h:\R \setminus \left(\Z+\frac{1}{2}\right)
 \to\R$ is $1$-periodic, i.e. $h(x+1)=h(x)$  for all $x \in \R \setminus \left(\Z+\frac{1}{2}\right)$, $h|_{\left(-\frac{1}{2},\frac{1}{2}\right)}$ is injective, and $(\omega, 1) \in \R^d \times \R$ is Diophantine: 
 \begin{equation}\label{eqn:diofnew}
\min_{j \in \Z}\, \left| \omega \cdot k  +j \right|\geq \frac{1}{C |k|^\gamma}, \quad \forall\, k \in \Z^d_0.
\end{equation}
\end{itemize}
The following result states some sufficient conditions under which  the sequence $(\lambda_n)_{n \in \Z^d}$ in~(\ref{lambdan})  satisfies  Assumptions~\ref{A.1new}--\ref{A.3new}.   
\begin{lemma} Let $a, b >0$.
\label{lemma1}
\begin{enumerate}
\item Let  $h \in C^{1}(\R)$ be such that
\begin{equation}
\inf_{x\in\R} |h'(x)| = a >0,
\label{eq:minfirst}
\end{equation}
and
\begin{equation}
\label{eq:maxsecond}
\left| \frac{\mathrm{d}}{\mathrm{d} x}\left(\frac{1}{h(x+y) - h(x) }\right)\right|  \leq b \left(1 + \frac{1}{|y|}\right),
\end{equation}
for $y\neq 0$, and let $\omega \in \R^d$ satisfy condition~(\ref{eq:mundiof}).
Then the sequence $(\lambda_n)_{n\in\Z}$  given by~\eqref{lambdan} fulfills Assumptions~\ref{A.1new}--\ref{A.3new}.
\item Let $h:\R \setminus\left(\Z+\frac{1}{2}\right)
\to\R$ be $1$-periodic, and $h|_{\left(-\frac{1}{2},\frac{1}{2}\right)} \in C^1\left(-\frac{1}{2}, \frac{1}{2}\right)$,  such that 
\begin{equation}
\inf_{x\in \left( -\frac{1}{2}, \frac{1}{2} \right)} |h'(x)| = a >0,
\label{eq:minfirst1}
\end{equation}
and
\begin{equation}
\label{eq:maxsecond1}
\left| \frac{\mathrm{d}}{\mathrm{d} x}\left(\frac{1}{h(x+y) - h(x) }\right)\right|  < b \left(1 + \frac{1}{|y|}\right),
\end{equation}
for $x, y \in \left( -\frac{1}{2}, \frac{1}{2} \right)$, with $y \neq 0$ and $|x+y| \neq \frac{1}{2}$, and let $\omega \in \R^d$ satisfy condition~(\ref{eqn:diofnew}).
Then the sequence $(\lambda_n)_{n\in\Z}$  given by~\eqref{lambdan} fulfills Assumptions~\ref{A.1new}--\ref{A.3new}.
\end{enumerate}
\end{lemma}
\begin{proof}(i)
For all $n,k \in \Z^d$, $k \neq 0$, we can write, by~\eqref{eq:minfirst}
\begin{equation}
\label{eq:infhp}
|\lambda_{n+k} - \lambda_n| = |h(\mu_{n+k}) - h(\mu_n)|  = |h(\mu_{n}+\mu_k) - h(\mu_n)|  \geq |\mu_k| \inf_{x\in\R}|h^\prime(x)| = a |\mu_k|.
\end{equation}
Thus, by~\eqref{eq:mundiof}
$$
\frac{1}{|\lambda_{n+k}-\lambda_n|} \leq \frac{C}{a} |k|^{\gamma},
$$
that is Assumption~\ref{A.1new}. Moreover, one gets
$$
|\lambda_{k}-\lambda_0|= |h'(\theta)| |\mu_k|,
$$
for some $|\theta|\leq |\mu_k|$. 
Hence,
\begin{eqnarray}
&&
\nonumber
\frac{1}{|\mu_k|} = \frac{|h'(\theta)|}{|\lambda_k-\lambda_0|}
\leq \frac{\max_{|x|\leq 1} |h'(x)|}{|\lambda_k-\lambda_0|} , \quad\textrm{if $ |\mu_k|\leq 1$}
\\
&&
\nonumber
\frac{1}{|\mu_k|} < 1 \leq \left( 1 + \frac{1}{|\lambda_k-\lambda_0|}\right), \quad\textrm{if $ |\mu_k|> 1$},
\end{eqnarray}
that is,
\begin{equation}
\label{eq:majmuk}
\frac{1}{|\mu_k|}  \leq \Big(1+ \max_{|x|\leq 1} |h'(x)|\Big)  \left( 1 + \frac{1}{|\lambda_k-\lambda_0|}\right).
\end{equation}
Finally, by~\eqref{eq:infhp} and~\eqref{eq:majmuk}, we have:
\begin{equation}
\frac{1}{|\lambda_{n+k}-\lambda_n|} \leq \frac{1}{a |\mu_k|} 
 \leq \frac{1}{a} \Big(1+ \max_{|x|\leq 1} |h'(x)|\Big)  \left( 1 + \frac{1}{|\lambda_k-\lambda_0|}\right), 
\label{eq:A.2.1dp}
\end{equation}
that is Assumption~\ref{A.2new}.
Concerning Assumption~\ref{A.3new}, for all $n,j,k \in \Z^d$, with $k \neq 0$, we  have
\begin{eqnarray*}
R_{n,j,k} &=&  \frac{1}{\lambda_j-\lambda_0}\left(\frac{1}{\lambda_{n+j+k}-\lambda_{n+j} } - \frac{1}{\lambda_{n+k}-\lambda_n} \right)
\\
 &=&  \frac{1}{h(\mu_j)-h(0)} \left(\frac{1}{h(\mu_n+\mu_j +\mu_k)-h(\mu_n+\mu_j ) } - \frac{1}{h(\mu_n+\mu_k)-h(\mu_n)} \right)
\\
&=&  \frac{1}{h(\mu_j)-h(0)} \int_{\mu_n}^{\mu_n + \mu_j} \frac{\mathrm{d}}{\mathrm{d}x} \left(\frac{1}{h(x+\mu_k)-h(x) }  \right) \mathrm{d}x.
\end{eqnarray*}
Therefore, \eqref{eq:infhp},   \eqref{eq:maxsecond}, and~\eqref{eq:majmuk} yield
$$
|R_{n,j,k}| \leq   \frac{b}{a}  \left(1+ \frac{1}{|\mu_k|}\right) \leq 
  \frac{b}{a} \Big(2+ \max_{|x|\leq 1} |h'(x)|\Big)  \left( 1 + \frac{1}{|\lambda_k-\lambda_0|}\right).
$$
Hence Assumptions~\ref{A.1new}--\ref{A.3new} hold with the constant 
$$
c= \frac{1+ b + C}{a} 
 \Big(2 + \max_{|x|\leq 1} |h'(x)|\Big).
$$

(ii) First observe that for all $x \in \R$ there is a unique $j_x \in \Z$ such that $x +j_x \in \left( -\frac{1}{2}, \frac{1}{2}\right]$ and we denote $w(x):= x+j_x$. Let $k \in \Z^d_0$. We observe that
$$
\frac{1}{|\lambda_k-\lambda_0|}=\frac{1}{|h (\mu_k)- h(\mu_0)|}=\frac{1}{|h (w(\mu_k))- h(0)|}=\frac{1}{|w(\mu_k)h'(\eta)|},
$$
with $|\eta| \leq |w(\mu_k )|$. Moreover for all $n \in \Z^d$:
$$
\frac{1}{|\lambda_{n+k}-\lambda_n|}=\frac{1}{|h (\mu_n+\mu_k)- h(\mu_n)|}=\frac{1}{|h (w(\mu_n+\mu_k))- h(w(\mu_n))|}=\frac{1}{|(w(\mu_k)+m)h'(\eta_1)|}
$$
for some $\eta_1 \in \left(-\frac{1}{2}, \frac{1}{2}\right)$ and $m \in\Z$, with $|m| \leq 1$. Thus by~(\ref{eqn:diofnew}) we obtain
$$
\frac{1}{|\lambda_{n+k}-\lambda_n|}=\frac{1}{|(w(\mu_k)+m)h'(\eta_1)|}\leq \frac{1}{\min_{j \in \Z}|\omega \cdot k + j||h'(\eta_1)|}\leq \frac{C}{a}|k|^\gamma,
$$
that is Assumption~\ref{A.1new}.
Moreover, since $|w(\mu_{k})|< \frac{1}{2}$ we have that
$$
\frac{1}{|w(\mu_k)+m|} \leq \frac{1}{|w(\mu_k)|}.
$$
Now, for all $0< \delta <\frac{1}{2}$ we define 
$$
a_{\delta}:= \max_{|x|\leq \delta} |h'(x)|
$$
and distinguish two cases: if $|w(\mu_k)|> \delta$ then 
$$
\frac{1}{|w(\mu_k)|} < \frac{1}{\delta},
$$
while if $|w(\mu_k)| \leq \delta$ then 
$$
\frac{1}{|w(\mu_k)|} = \frac{|h'(\eta)|}{|\lambda_k-\lambda_0|}\leq \frac{a_\delta}{|\lambda_k-\lambda_0|},
$$
so that
\begin{equation}
	\frac{1}{|w(\mu_k)|} < A_\delta \left( 1+ \frac{1}{|\lambda_k-\lambda_0|}\right),
	\label{eq:3.12}
\end{equation}
where 
$$
A_\delta:= \max \left\{ \frac{1}{\delta }, a_\delta \right\}.
$$
Therefore,
$$
\frac{1}{|\lambda_{n+k}-\lambda_n|}=\frac{1}{|(w(\mu_k)+m)h'(\eta_1)|}\leq\frac{1}{|w(\mu_k)h'(\eta_1)|} \leq \frac{A_\delta}{a}\left( 1+ \frac{1}{|\lambda_k-\lambda_0|}\right),
$$
that is Assumption~\ref{A.2new}. 

Concerning Assumption~\ref{A.3new}, let $n,k, j \in \Z^d$, $j,k \neq 0$, and consider
\begin{eqnarray}
&&\left| \frac{1}{\lambda_j -\lambda_0}\left( \frac{1}{\lambda_{n+j+k}- \lambda_{n+j}} - \frac{1}{\lambda_{n+k}- \lambda_{n}}\right) \right| \nonumber \\
&& = \left| \frac{1}{h(\mu_j) -h(\mu_0)}\left( \frac{1}{h(\mu_n+\mu_j+\mu_k)- h(\mu_n+\mu_j)} - \frac{1}{h(\mu_n+\mu_k)- h(\mu_n)}\right) \right| \nonumber \\
&& = \left| \frac{1}{h(y_j)-h(0)}\left( \frac{1}{h(x_n+y_j+z_k)- h(x_n+y_j)} - \frac{1}{h(x_n+z_k)- h(x_n)}\right) \right| \nonumber 
\end{eqnarray}
where $x_n=w(\mu_n), y_j=w(\mu_j), z_k=w(\mu_k)$. Notice that $x_n,y_j, z_k \in \left( -\frac{1}{2}, \frac{1}{2}\right)$, $y_j,z_k \neq 0$ and $|w(x_n+y_j)|, |w(x_n+y_j+z_k)|, |w(x_n+z_k)| \neq \frac{1}{2} $. Consider the function
$$
G(x,y,z)=\frac{1}{h(y)-h(0)}\left( \frac{1}{h(x+y+z)- h(x+y)} - \frac{1}{h(x+z)- h(x)}\right),
$$
on its domain. By~(\ref{eq:maxsecond1}) we have
$$
\lim_{y \to 0}|G(x,y,z)|=\left| \frac{1}{h'(0)} \frac{d}{d x}\left( \frac{1}{h(x+z)-h(x)}\right)\right| < \frac{b}{a}\left(1+\frac{1}{|z|}\right).
$$
Thus there is a $0<\delta_1 <\frac{1}{2}$ such that if $|y_j| \leq \delta_1$ then 
$$
|G(x,y,z)| \leq \frac{b}{a}\left(1+\frac{1}{|z|}\right) .
$$
Thus if $|y_j| \leq \delta_1$, then 
$$
|G(x_n,y_j,z_k)| \leq \frac{b}{a}\left(1+\frac{1}{|z_k|}\right) = \frac{b}{a}\left(1+\frac{1}{|w(\mu_k)|}\right) \leq \frac{b}{a}\left( A_{\delta_1}+1 \right)\left(1+\frac{1}{|\lambda_k-\lambda_0|}\right) ,
$$
by~\eqref{eq:3.12}. Moreover, since 
$$
\frac{1}{|\lambda_{n+j+k}-\lambda_{n+j}|}\leq \frac{A_{\delta_1}}{a}\left(1+\frac{1}{|\lambda_k-\lambda_0|}\right)
$$
and
$$
\frac{1}{|\lambda_{n+k}-\lambda_{n}|}\leq\frac{A_{\delta_1}}{a}\left(1+\frac{1}{|\lambda_k-\lambda_0|}\right),
$$
we have that if $|y_j| > \delta_1$ then $|G(x_n,y_j,z_k)|$ can be estimated as follows:
\begin{eqnarray}
|G(x_n,y_j,z_k)|&=&\frac{1}{|y_j| |h'(\zeta)|}\left|\frac{1}{\lambda_{n+j+k}-\lambda_{n+j}}- \frac{1}{\lambda_{n+k}-\lambda_{n}}\right| \nonumber  \\
              & \leq & \frac{2 A_{\delta_1}}{\delta_1 a^2}\left(1+\frac{1}{|\lambda_k-\lambda_0|}\right),
\end{eqnarray}
where $|\zeta| \leq |y_j|$.
Therefore we have that
$$
\left| \frac{1}{\lambda_j -\lambda_0}\left( \frac{1}{\lambda_{n+j+k}- \lambda_{n+j}} - \frac{1}{\lambda_{n+k}- \lambda_{n}}\right) \right| \leq  \max\left\{ \frac{b(A_{\delta_1}+1)}{a}, \frac{2 A_{\delta_1}}{\delta_1a^2}\right\}\left(1+\frac{1}{|\lambda_k-\lambda_0|}\right).
$$

We conclude that $(\lambda_n )_{n \in \Z^d}$ satisfies Assumptions~\ref{A.1new}--\ref{A.3new} with the constant 
$$
c=\max\left\{\frac{A_{\delta_1}}{a}, \frac{b(A_{\delta_1}+1)}{a}, \frac{2 A_{\delta_1}}{\delta_1a^2},\frac{C}{a}\right\}.
$$
\end{proof}
\begin{example}  Examples generated by~\eqref{lambdan}  through Lemma~\ref{lemma1} (elementary verification of~\eqref{eq:minfirst},  \eqref{eq:maxsecond},  \eqref{eq:minfirst1},  \eqref{eq:maxsecond1}  omitted).
\begin{enumerate}
\item 
$h(x) = x + \beta x^3$, with $\beta\geq 0$;
\item 
$h(x)= \tan( \pi x)$;  
\item For all $x \in \R$ there is a unique $j_x \in \Z$ such that $x +j_x \in \left( -\frac{1}{2}, \frac{1}{2}\right]$, let $h(x)=x+j_x$.
\end{enumerate}
\end{example}
Concerning examples of perturbations $V$ fulfilling Assumption~\ref{A.4new}, a  family of them is produced by the following Lemma.
\begin{lemma}
\label{lemma2}
Let $(\lambda_n)_{n \in \Z^d} $ be a sequence as in~\eqref{lambdan} with  $h$ as in Lemma~\ref{lemma1}. Let $(f_k)_{k \in\Z^d}$ be a sequence of complex functions satisfying one of the following two properties:
\begin{enumerate}
\item there is $M>0$ such that $f_k \equiv 0$ if $|k| >M$, and for all $k \in \Z^d$:
$$
\|f_k \circ h\|_{C^1}:=\sup_{x} |f_k (h(x))| + \sup_{x} \left|  \frac{d}{d x} f_k(h(x))\right| < + \infty.
$$
\item there is $\alpha>0$ such that $f_k=e^{-\alpha |k|} g_{k}$, for all $k \in \Z^d$, where $(g_k)_{k \in \Z^d}$ is a sequence of complex functions such that 
\begin{equation}
\sup_{k \in \Z^d}\|g_k \circ h\|_{C^1}< + \infty.
\end{equation}
\end{enumerate}
Then the operator $V \in \mathcal{L}(\ell^2(\Z^d))$ with matrix representation 
\begin{equation}
\label{eq:Vshape}
V_{m,m+k}=f_k(\lambda_m), \quad k,m \in \Z^{d},
\end{equation}
satisfies  Assumption~\ref{A.4new}.
\end{lemma}

\begin{proof} 
We have that for all $k \in Z^d$:
\begin{eqnarray}
&& \sup_{m \in \Z^d} |V_{m,m+k}|+\sup_{m \in \Z^d, \,j \in \Z^d_0} \left|\frac{V_{m+j,m+j+k}-V_{m,m+k}}{\lambda_j -\lambda_0}\right| \nonumber \\
&& =  \sup_{m \in \Z^d} |f_k(h(\mu_m))|+\sup_{m \in \Z^d, \,j \in \Z^d_0} \left|\frac{f_k(h(\mu_{m+j}))-f_k(h(\mu_m))}{h(\mu_j) -h(0)}\right| \nonumber \\
&& \leq   \sup_{x} |f_k(h(x))|+\sup_{\stackrel{x,y}{y \neq 0}} \left|\frac{f_k(h(x))-f_k(h(y))}{h(y) -h(0)}\right|. \nonumber \\
&& \leq \left(1+\frac{1}{a}\right) 
\|f_k \circ h\|_{C^1} < + \infty,
\end{eqnarray}
with $a=\inf_{x}|h'(x)|>0$. Therefore if $(f_k)_{k \in\Z^d}$ satisfies one of the conditions $(i)$, $(ii)$, then  Assumption~\ref{A.4new} holds. 
\end{proof}
\begin{remark} 
 Notice that if $V=V^*$ the functions $(f_k)_{k \in \Z^d}$ cannot be independently chosen, but have to satisfy the relation
 $$
 f_{-k}(\lambda_m) = \overline{f_k(\lambda_{m-k})}, \qquad \textrm{for all $k,m \in \Z^{d}$}.
 $$
The independent functions are thus the set $(f_k)_{k \in \N \times \Z^{d-1}}$.
\end{remark}
\begin{remark}
The case $h(x)=\tan (\pi x)$ and  $V$ the discrete Laplacian on $\ell^2(\Z^d)$, namely, for all $m, k \in \Z^d$,
$$ 
V_{m,m+k} = \delta_{|k|,1},
$$
is the example considered in~\cite{ref:Poeschel, ref:Bellissard, ref:Fishman}.
\end{remark}
\section{Road map of the proof of Theorem~\ref{thm:mainth2}}\label{sect:roadmap}
We recall that $T$ is a self-adjoint multiplication operator on $\ell^2(\Z^d)$, $d\geq 1$, with  purely point spectrum with distinct eigenvalues $(\lambda_n)_{n\in\Z^d}$,  
satisfying Assumptions~\ref{A.1new}--\ref{A.3new}. The operator $V$ is a bounded  operator on $\ell^2(\Z^d)$ satisfying Assumption~\ref{A.4new}. We want to prove that there exists  $\e^* >0$ (depending on $T$ and $V$) such that for all $\e \in [-\e^*,\e^*]$ the spectrum of $T+\e V$ is purely point with simple eigenvalues $( \lambda_n(\e))_{n \in \Z^d}$ such that 
$$
\frac{1}{|\lambda_n(\e)-\lambda_m(\e)|} \leq 12c^2|n-m|^{2\gamma},  \quad \forall n,m \in \Z^d, n\neq m.
$$

\subsection{KAM iteration scheme}
The proof of Theorem~\ref{thm:mainth2} is based on a KAM-type iterative procedure that we briefly recall here. 

\subsubsection{First step}
In the first step of the procedure we want to find an operator $U^{(1)}$  such that  
\begin{equation}\label{eqn:U1}
{U^{(1)}}^{-1}(T+\e V)U^{(1)}=T^{(1)}+ \e^2 V^{(1)},
\end{equation}
where 
\begin{equation}
	T^{(1)}=T+\e [V], \qquad  \e^2 V^{(1)}={U^{(1)}}^{-1}(\e V- \e [V])(U^{(1)}-I).
\label{eq:4.2}
\end{equation}
 Here, given a bounded operator $A$ on $\ell^2(\Z^d)$  its diagonal part $[A]$ is defined by $[A]_{n,m}= \delta_{m,n} A_{n,m}$, for all $n,m \in \Z^d$.

With the transformation~\eqref{eqn:U1} we obtain a new operator $T^{(1)}+ \e^2 V^{(1)}$, where $T^{(1)}$ is a multiplication (diagonal) operator with eigenvalues $\lambda^{(1)}_n= \lambda_n+ \e  V_{n,n} $, $n \in \Z^d$, and a new perturbation $\e^2 V^{(1)}$ of order $\e^2$.

\subsubsection{Homological equation}

The transformation $U^{(1)} $in~(\ref{eqn:U1}) is the solution of the \emph{homological equation}
\begin{equation}\label{eqn:homeq}
\left[T^{(1)}, U^{(1)}\right]+\e V-\e \left[V\right]=0.
\end{equation}
In fact, one has
\begin{eqnarray*}
U^{(1)}\left(T^{(1)}+ \e^2 V^{(1)}\right){U^{(1)}}^{-1}&=&\left(-\left[T^{(1)}, U^{(1)}\right]+T^{(1)}U^{(1)}\right){U^{(1)}}^{-1}+ U^{(1)}\e^2 V^{(1)}{U^{(1)}}^{-1}\\
    &=& T+\e V,
\end{eqnarray*}
and by using~\eqref{eq:4.2}, one gets~\eqref{eqn:homeq}.

The homological equation~(\ref{eqn:homeq}) admits a one-parameter family of solutions  whenever $T$ has simple eigenvalues, namely, for all $n,m \in \Z^d$, $n \neq m$:
$$
  U^{(1)}_{n,m} = \frac{V_{n,m} }{\lambda_n-\lambda_n},
$$
while the diagonal elements $U^{(1)}_{n,n}=C\in\R$ for all $n\in \Z^d$. 

\subsubsection{Iteration}
The  KAM iteration proceeds as follows: given the initial conditions $T^{(0)}=T$, $V^{(0)}=V$, $U^{(0)}=I$, $T^{(1)}=T^{(0)}+ \e \left[V^{(0)}\right]$,
at each step $\ell \in \N$ of the scheme we want to find two operators $U^{(\ell+1)}, V^{(\ell+1)} \in {\mathcal L}\left(\ell^2\left(\Z^d\right)\right)$ such that
\begin{equation}\label{eqn:Ul1TUlroad}
{U^{(\ell+1)}}^{-1}(T+ \e V)U^{(\ell+1)}=T^{(\ell+1)}+ \e_{\ell+1} V^{(\ell+1)},
\end{equation}
where
$$
T^{(\ell+1)}=T+ \sum_{j=0}^{\ell} \e_j \left[V^{(j)}\right]
$$
and $\e_j:=\e^{2^j}$ for all $j=0, \dots, \ell+1$.
To do this we have to prove that  the homological equation
\begin{equation}
\label{hom01road}
\left[T^{(\ell+1)}, W^{(\ell)}\right]+ \e_\ell V^{(\ell)}-\e_\ell \left[V^{(\ell)}\right]=0,
\end{equation}
admits an invertible solution ${W^{(\ell)}}$ and then define
$$
\e_{\ell+1} V^{(\ell+1)}:= {W^{(\ell)}}^{-1}\left(\e_\ell V^{(\ell)}-\e_\ell \left[V^{(\ell)}\right]\right)\left(W^{(\ell)}-I\right),
$$
and
$$
U^{(\ell+1)}:=U^{(\ell)}W^{(\ell)}.
$$

\subsection{Strategy of the proof}
Here we present the road map of the proof of Theorem~\ref{thm:mainth2}: 
\begin{enumerate}
\item[(1)] In the first part of the proof we will prove that the aforementioned KAM scheme can be iterated. More precisely, at each step $\ell \in \N$ of the scheme, we have to solve the homological equation~(\ref{hom01road}) and hence we have to prove that the operator 
$$
T^{(\ell+1)}=T+ \sum_{j =0}^{\ell} \e_{j} \left[V^{(j)}\right],
$$
has simple eigenvalues: this is the most difficult part of the proof. The key ingredient is Theorem~\ref{thm:propdiof}  
which essentially says that, if the operators $V^{(0)}, V^{(1)}, \dots, V^{(\ell)}$ are bounded in some suitable norms and $\e$ is small enough then the perturbation
\begin{equation} \label{eqn:perturbl}
\sum_{j =0}^{\ell} \e_j \left[V^{(j)}\right]
\end{equation}
is such that $T^{(\ell+1)}$ has a purely point spectrum with simple eigenvalues $\left( \lambda^{(\ell+1)}_n\right)_{n \in \Z^d}$ satisfying a Diophantine condition~(\ref{eqn:estDiofantea}) and the homological equation~(\ref{hom01road}) admits solution. 
In Subsection~\ref{sect:notation} we define  the aforementioned norms for operators on $\ell^2(\Z^d)$  and it will be immediately clear the meaning of Assumption~\ref{A.4new}.
\item[(2)] In the second part of the proof we will prove that the sequence of transformations $\left( U^{(\ell)} \right)_{\ell \in \N}$ converges to an invertible operator $U \in {\mathcal L}(\ell^2(\Z^d))$ and that the series
$$
\sum_{j =0}^{+\infty} \e_j \left[V^{(j)}\right]
$$
is convergent, so that by passing to the limit $\ell \to + \infty$ in~(\ref{eqn:Ul1TUlroad}) we obtain
\begin{equation}\label{eqn:Uinf1}
U^{-1}(T+\e V)U=T^{(\infty)},
\end{equation}
where 
$$
T^{(\infty)}=T+\sum_{j =0}^{+\infty} \e_j \left[V^{(j)}\right].
$$
Moreover, we will prove that the eigenvalues of $T^{(\infty)}$, denoted $\left(\lambda^{(\infty)}_n\right)_{n \in \Z^d}$, satisfy the Diophantine condition
\begin{equation}\label{eqn:Diphest}
\frac{1}{\left|\lambda_n^{(\infty)}-\lambda_m^{(\infty)}\right|} \leq 12c^2|n-m|^{2\gamma},  \quad \forall n,m \in \Z^d, n\neq m.
\end{equation}
Therefore, by~(\ref{eqn:Uinf1}), $T(\e)=T+\e V$ has the same eigenvalues of $T^{(\infty)}$, i.e. $\lambda_n(\e)=\lambda^{(\infty)}_n$, for all $n \in \Z^d$, and they  satisfy~(\ref{Diofanto3}). 
\item[(3)] In the third part of the proof we will prove that if $V=V^*$ then  for all $n,m \in \Z^d$:
$$
\langle e_n, U^* U e_m \rangle = \delta_{n,m} \| U e_n\|^2,
$$
hence $U^\ast U$ is a strictly positive operator that commutes with $T^{(\infty)}$ so that  $\tilde{U}=U(U^\ast U)^{-1/2}$ is unitary and
\begin{equation}
\tilde{U}^{-1}(T+\e V)\tilde{U}=T^{(\infty)}.
\end{equation}
Moreover, the orthonormal eigenvectors of $T(\e)$ are
$$
u_n(\e)=\tilde{U}e_n=\frac{Ue_n}{\| Ue_n\|}, \qquad n \in \Z^d.
$$
Finally, by using some properties of the operator $U$ we will obtain the eigenvector decay property~(\ref{eqn:exploc}).
\end{enumerate}

\section{Notation and stability of small denominators}
\subsection{Notation}\label{sect:notation}
First of all we introduce the algebra of sequences $\mathcal{M}_T$ and the shift operator.
\begin{definition}
\label{defn:sequence}
A complex sequence $ a :=(a_n)_{n \in \Z^d}$ belongs to $\mathcal{M}_T$ if
$$
\| a \|_{T}:= \sup_{n \in \Z^d} |a_{n}|+\sup_{n \in \Z^d, \,j \in \Z^d_0} \left|\frac{a_{n+j}-a_{n}}{\lambda_j -\lambda_0}\right|<+\infty.
$$
\end{definition}
  
\begin{definition}
Let $ a =(a_n)_{n\in\Z^d}$ be a complex sequence and $k \in \Z^d$. The shift operator $\Theta_k$ is defined as:
\begin{equation}
\label{shift}
\Theta_k a :=(a_{n+k})_{n\in\Z^d}.
\end{equation}
\end{definition}

\begin{lemma}\label{lemma:contrnorm}
The following statements hold:
\begin{enumerate}
\item The constant sequence ${\bf 1}:=( 1 )_{m\in \Z^d} \in \mathcal{M}_T$ and $ \| {\bf 1} \|_{T}=1$;
\item $\forall\, a=(a_n)_{n \in \Z^d}  \in \mathcal{M}_T$: $ \sup_{n \in \Z^d} |a_n| \leq \|  a  \|_{T}$;
\item $\forall\, a=(a_n)_{n \in \Z^d}  \in \mathcal{M}_T$ and $\forall\,k \in \Z^d$:  $\Theta_k a \in \mathcal{M}_T$, and $\|\Theta_k  a  \|_{T}= \|  a  \|_{T}$;
\item $\mathcal{M}_T$ is a Banach algebra, i.e.\ for all $a=(a_m)_{m \in \Z^d}, b=( b_m)_{m \in \Z^d} \in \mathcal{M}_T$ one has that $a \, b:=(a_m b_m)_{m \in \Z^d}$ belongs to $\mathcal{M}_T$ and
$$
\| a \, b\|_{T} \leq \| a \|_{T} \| b \|_{T}.
$$
\end{enumerate}
\end{lemma}
\begin{proof}
We omit the elementary proof of $(i), (ii),$ and $(iii)$, we only give the proof of $(iv)$. Let  $a=(a_m)_{m \in \Z^d}, b=( b_m)_{m \in \Z^d} \in \mathcal{M}_T$, then
\begin{eqnarray}
\| a \, b\|_{T} & = & \sup_{n \in \Z^d}|a_n b_n|+\sup_{n \in \Z^d, \,j \in \Z^d_0} \left|\frac{a_{n+j}b_{n+j}-a_{n}b_n}{\lambda_j -\lambda_0}\right| \nonumber \\
& \leq  &  \sup_{n \in \Z^d}|a_n b_n|+  \sup_{k \in \Z^d}|a_k |\sup_{n \in \Z^d, \,j \in \Z^d_0} \left|\frac{b_{n+j}-a_{n}}{\lambda_j -\lambda_0}\right| \nonumber \\
& & + \sup_{k \in \Z^d}|b_k |\sup_{n \in \Z^d, \,j \in \Z^d_0} \left|\frac{a_{n+j}-a_{n}}{\lambda_j -\lambda_0}\right| \nonumber \\
& \leq &  \| a \|_{T} \| b \|_{T}. \nonumber
\end{eqnarray}
\end{proof}

\begin{definition}
For any $A \in {\mathcal L}(\ell^2(\Z^d))$ we consider its matrix elements with respect to the canonical basis $(e_n)_{n\in\Z^d}$ of $\ell^2(\Z^d)$: 
$$
A_{n,m}:= \langle e_n, A e_m \rangle, \qquad n,m \in \Z^d,
$$ 
and its main diagonals: 
$$
A_k:=( A_{n,n+k} )_{n \in \Z^d}, \qquad k \in \Z^d.
$$
\end{definition}
Definition~\ref{defn:sequence} in turn motivates the introduction of the  following family of subspaces of ${\mathcal L\left(\ell^2(\Z^d)\right)}$.
\begin{definition}\label{def:Malpha}
For all $\alpha >0$ we define
$$
 M^{\alpha}:=\left\{A \in {\mathcal L}\left(\ell^2(\Z^d)\right) \, :\,\
   \|A\|_{\alpha}:= \sup_{k \in \Z^d} e^{\alpha |k|}\,\| A_k\|_{T}  <+ \infty \right\}.
$$
Moreover $M^\infty$ is the space of the diagonal operators $A$ 
such that $\|A\|_{\infty}:=\|A_0\|_{T} < +\infty$. We denote with $A \in M^\alpha \mapsto [A] \in M^\infty$  the canonical projection of $M^{\alpha}$ onto $M^{\infty}$.
\end{definition}
\begin{remark}
$M^\alpha$  is a Banach space under the norm $\| \cdot\|_\alpha$ for any $\alpha >0 $ or $\alpha=\infty$.  Definition~\ref{defn:sequence} relates $\| \cdot\|_\alpha$ to the eigenvalues $(\lambda_n)_{n \in \Z^d}$ of the operator $T$.
Using this notation Assumption~\ref{A.4new} can be reformulated as follows: 
\begin{equation}\label{eqn:decayingk}
\|V\|_{\alpha}=\sup_{k \in \Z^d} e^{|k|\alpha} \| V_k\|_{T}  < +\infty,
\end{equation}
for some $\alpha>0$.
\end{remark}
\begin{lemma}\label{lemma:ineqnorm}
The following assertions hold:
\begin{enumerate}
\item If $\alpha > \beta >0$ then $M^{\alpha} \subset M^{\beta}$ and $\| A\|_{\alpha} > \| A \|_{\beta}$ for all $A \in M^{\alpha}$; 
\item If $\alpha >0$ then $M^{\infty} \subset M^{\alpha}$ and $\|B\|_{\infty}=\| B \|_{\alpha}$ for all $B \in M^{\infty}$;
\item If $\alpha >0$ then $\| [C] \|_{\infty} \leq \| C \|_{\alpha}$ and $\| (C - [C])  \|_{\alpha} \leq \| C \|_{\alpha}$ for all $C \in M^{\alpha}$; 
\item The map $ a= ( a_n )_{n \in \Z^d}  \in \mathcal{M}_T \mapsto A \in M^\infty$, with $A_{n,m}=a_n \delta_{n,m}$ for all $n,m \in \Z^d$, is an isometry between $\mathcal{M}_T$ and $M^\infty$;
\item If $A \in M^{\alpha}$, $\alpha >0$, then 
\begin{equation}
\label{qalpha}
\| A\| \leq  \left( \frac{1+ e^{-\alpha}}{1- e^{-\alpha}} \right)^d \|A\|_{\alpha},
\end{equation}
where $\| A\|$ is the operator norm of $A$ in ${\mathcal L}\left(\ell^2(\Z^d)\right)$. If $A \in M^{\infty}$ then $ \|A\| \leq \|A\|_{\infty}$.
\end{enumerate}
\end{lemma}
\begin{proof}
The proof of $(i), (ii), (iii)$ and $(iv)$ are straightforward, we only give the proof of $(v)$. If $A \in M^{\alpha}$, $\alpha >0$, then for $v\in \ell^2(\Z^d)$,
\begin{eqnarray}
\|A v\|^2 &=& \sum_{m \in \Z^d}\Bigl|\sum_{j \in \Z^d} A_{m,j} v_j\Bigr|^2  \nonumber \\
&\leq& \sum_{m \in \Z^d}\Bigl(\sum_{k \in \Z^d} |A_{m,m+k}| \,|v_{m+k}|\Bigr)^2  
\nonumber\\
&=& \sum_{m \in \Z^d}\Bigl(\sum_{k \in \Z^d} |(A_{k})_m| \,|v_{m+k}|\Bigr)^2  
\nonumber\\
&\leq& \sum_{m \in \Z^d}\Bigl(\sum_{k \in \Z^d} \|A_{k}\|_{T}\, |v_{m+k}|\Bigr)^2  
\label{cfr_norm}\\
&=& \sum_{m \in \Z^d}\Bigl(\sum_{k \in \Z^d} \|A_{k}\|_{T}\, e^{\alpha |k|} e^{-\alpha |k|} \,|v_{m+k}|\Bigr)^2  
\nonumber\\
&\leq& \|A\|_{\alpha}^2 \sum_{m \in \Z^d}\Bigl(\sum_{k \in \Z^d}   e^{-\alpha |k|} \,|v_{m+k}|\Bigr)^2  
\nonumber\\
&=& \|A\|_{\alpha}^2 \sum_{k,j \in \Z^d}
   e^{-\alpha |k|} e^{-\alpha |j|} \sum_{m \in \Z^d}|v_{m+k}|\,|v_{m+j}| 
\nonumber\\
&\leq& \|A\|_{\alpha}^2 \sum_{k,j \in \Z^d}
   e^{-\alpha |k|} e^{-\alpha |j|} \sum_{m \in \Z^d}|v_{m}|^2\, 
\nonumber\\
&=& \|A\|_{\alpha}^2 \Bigl(\sum_{k \in \Z^d}
   e^{-\alpha |k|} \Bigr)^2 \|v\|_{\ell^2(\Z^d)}^2 
\nonumber\\
      &=& \|A\|_{\alpha}^2 \left( \frac{1+ e^{-\alpha}}{1- e^{-\alpha}} \right)^{2d} \|v\|_{\ell^2(\Z^d)}^2, \label{cfr_norm1}
\end{eqnarray}
whence one has~\eqref{qalpha}, where in~(\ref{cfr_norm}) we used the inequality  $\sup_{n \in \Z^d} |a_n| \leq \| a \|_{T}$, and in~\eqref{cfr_norm1} we used the equality
\begin{equation}
	\sum_{k\in \Z} e^{-\alpha |k|} = 2\sum_{k\in \N} e^{-\alpha k}-1 = \frac{2}{1-e^{-\alpha}}-1=\frac{1+e^{-\alpha}}{1-e^{-\alpha}}. 
\label{eq:sumk}
\end{equation}
Finally, if $A \in M^{\infty}$ then for all $\alpha >0$: $\|A\|_{\infty}= \|A\|_{\alpha}$, therefore
$$
\|A\| \leq \inf_{\alpha >0}  \left( \frac{1+ e^{-\alpha}}{1- e^{-\alpha}} \right)^d \|A\|_{\alpha}=\inf_{\alpha >0}  \left( \frac{1+ e^{-\alpha}}{1- e^{-\alpha}} \right)^d \|A\|_{\infty}=\|A\|_{\infty},
$$ 
because $q(\alpha)= \left( \frac{1+ e^{-\alpha}}{1- e^{-\alpha}} \right)^d$ is a strictly decreasing function and 
$$
\lim_{\alpha \to +\infty}q(\alpha)=\inf_{\alpha >0} q(\alpha)=1.
$$
\end{proof}

\subsection{Stability of small denominators}
\label{stabilityden}
Let us conclude the preliminaries by stating and proving a result that will play  a fundamental role in the subsequent KAM argument.
\begin{theorem}\label{thm:propdiof}
Let $T$ be a self-adjoint operator acting on $\ell^2(\Z^d)$, $d\geq 1$, with  purely point spectrum with eigenvalues $(\lambda_n)_{n\in\Z^d}$ corresponding to the canonical basis $( e_n )_{n \in \Z^d}$ satisfying Assumptions~\ref{A.1new}--\ref{A.3new}.  Then the following assertions hold:
\begin{enumerate}
\item Let $v=(v_n)_{n \in \Z^d} \in \mathcal{M}_T$ and $\tilde{\lambda}=(\tilde{\lambda}_n)_{n \in \Z^d}$ be the sequence with $\tilde{\lambda}_n=\lambda_n+v_n$, $n \in \Z^d$. If $\|v\|_{T} \leq 1/(4c)$, then for all $k \in \Z^d_{0}$ the sequence 
$$
\left(\Theta_k \tilde{\lambda}-\tilde{\lambda}\right)^{-1} =\left( \frac{1}{\lambda_{n+k}+v_{n+k}- \lambda_n -v_n}\right)_{n \in \Z^d} 
$$
belongs to $ \mathcal{M}_T$ and fulfills the estimate
$$
\left\| \left(\Theta_k \tilde{\lambda}-\tilde{\lambda}\right)^{-1}  \right\|_{T} \leq 12c^2 |k|^{2\gamma}.
$$
\item Let $\ell \in \N$, $\alpha_0 \geq \alpha_1 \geq \dots \geq \alpha_\ell >0$, and $V^{(0)}, V^{(1)}, \dots ,V^{(\ell)}$ be operators on $\ell^2(\Z^d)$ with $V^{(j)} \in M^{\alpha_j}$, for all $j=0,1, \dots, \ell$. Let $\e \in \R$ and 
\begin{equation}
T^{(\ell+1)}=T+ \sum_{j=0}^{\ell} \e_j \left[V^{(j)}\right],
\end{equation}
with $\e_j:=\e^{2^j}$.
If 
\begin{equation}\label{eqn:hypsumeps}
\sum_{j=0}^{\ell} \left\| \e_{j} V^{(j)} \right\|_{\alpha_j} \leq \frac{1}{4c},
\end{equation}
then $T^{(\ell+1)}$ is an operator with  purely point spectrum with eigenvalues $\left(\lambda^{(\ell+1)}_n\right)_{n\in\Z^d}$ corresponding to the canonical basis $(e_n )_{n \in \Z^d}$  and for all $k \in \Z^d_{0}$
\begin{equation}\label{eqn:estDiofantea}
 \sup_{n \in \Z^d}  \frac{1}{\left| \lambda^{(\ell+1)}_{n+k}- \lambda_n^{(\ell+1) }\right|} \leq  12c^2 |k|^{2\gamma}.
\end{equation}
Moreover, the homological equation
\begin{equation}\label{eqn:homologicalth}
\left[ T^{(\ell+1)}, W^{(\ell)} \right]+\e_{\ell} V^{(\ell)}- \e_\ell \left[V^{(\ell)}\right]=0,
\end{equation}
admits a unique solution $W^{(\ell)} \in M^{\alpha_\ell-\delta}$, $0< \delta <\alpha_\ell$, such that $\left[W^{(\ell)}-I\right]=0$ and
$$
\left\| W^{(\ell)}-I \right\|_{\alpha_\ell-\delta} < 12 c^2 \left(\frac{2 \gamma }{e \delta}\right)^{2 \gamma }  \left\|\e_\ell  V^{(\ell)}\right\|_{\alpha_\ell}.
$$
\end{enumerate}
\end{theorem}
\begin{proof}
We start by proving $(i)$. Let  $k \in \Z^d_0$, by definition
\begin{eqnarray}
& &\left\|  \left(\Theta_k \tilde{\lambda}-\tilde{\lambda}\right)^{-1} \right\|_{T}  =\sup_{n \in \Z^d}\left| \frac{1}{ \lambda_{n+k}+ v_{n+k}- \lambda_n- v_n } \right| \nonumber \\
&  & + \sup_{n \in \Z^d, \, j \in \Z_0^d} \left| \frac{1}{\lambda_j - \lambda_0} \left(\frac{1}{ \lambda_{n+j+k}+ v_{n+j+k}- \lambda_{n+j}- v_{n+j} } -\frac{1}{ \lambda_{n+k}+ v_{n+k}- \lambda_n- v_n }\right)
\right|. \nonumber \end{eqnarray}
Notice that, by Assumption~\ref{A.2new}, we have
\begin{eqnarray}
\left| \frac{v_{n+k}-v_n}{\lambda_{n+k}-\lambda_n}\right| & \leq & c |v_{n+k}-v_n|\left( 1 + \frac{1}{|\lambda_{k}-\lambda_0|}\right) \nonumber \\
& \leq & c \left( |v_{n+k}|+  |v_n|+ \frac{|v_{n+k}-v_{n}|}{|\lambda_k- \lambda_0|} \right) \nonumber \\
& \leq & 2c \|v\|_{T} \nonumber \\
& \leq & \frac{1}{2}. \nonumber
\end{eqnarray}
Therefore,
\begin{eqnarray}
 \frac{1}{ \left|\lambda_{n+k}+ v_{n+k}- \lambda_n- v_n  \right|}& = &  \frac{1}{ \left|\lambda_{n+k}- \lambda_n  \right| \left| 1+ \frac{v_{n+k}-v_n}{\lambda_{n+k}-\lambda_n}\right|}  \nonumber
 \\
   & \leq&   \frac{1}{ \left|\lambda_{n+k}- \lambda_n  \right| \left( 1- \left| \frac{v_{n+k}-v_n}{\lambda_{n+k}-\lambda_n} \right| \right) }  \nonumber \\
        & \leq & \frac{2}{ \left|\lambda_{n+k}- \lambda_n  \right|  }  \nonumber \\
      & \leq &2 c |k|^{\gamma}, \label{eq:useA.1}
\end{eqnarray}
where in~(\ref{eq:useA.1}) we used Assumption~\ref{A.1new}.
Moreover
\begin{eqnarray}
&& \left| \frac{1}{\lambda_j - \lambda_0}\left(
\frac{1}{ \lambda_{n+j+k}+ v_{n+j+k}- \lambda_{n+j}- v_{n+j} } -\frac{1}{ \lambda_{n+k}+ v_{n+k}- \lambda_n- v_n }\right)\right| \nonumber \\
&  & \leq  \frac{|\lambda_{n+k}- \lambda_n - \lambda_{n+j+k}+\lambda_{n+j}|}{|\lambda_j-\lambda_0| |\lambda_{n+j+k}-\lambda_{n+j}+v_{n+j+k}-v_{n+j}| |\lambda_{n+k}-\lambda_{n}+v_{n+k}-v_{n}|}  \nonumber
 \\
 &  &\quad  + \frac{|v_{n+j}- v_{n} |}{|\lambda_j-\lambda_0| |\lambda_{n+j+k}-\lambda_{n+j}+v_{n+j+k}-v_{n+j}| |\lambda_{n+k}-\lambda_{n}+v_{n+k}-v_{n}|}  \nonumber
 \\
 & &\quad  + \frac{|v_{n+k}- v_{n+j+k} |}{|\lambda_j-\lambda_0| |\lambda_{n+j+k}-\lambda_{n+j}+v_{n+j+k}-v_{n+j}| |\lambda_{n+k}-\lambda_{n}+v_{n+k}-v_{n}|}  \nonumber
 \\
 &  & \leq  \frac{1}{ \left( 1- 2c \| v \|_{T}\right)^2 }\left( \frac{|\lambda_{n+k}- \lambda_n - \lambda_{n+j+k}+\lambda_{n+j}|}{|\lambda_j-\lambda_0| |\lambda_{n+j+k}-\lambda_{n+j}| |\lambda_{n+k}-\lambda_{n}|} \right.  \nonumber
 \\
 &  &  \qquad\qquad\qquad\qquad+  \frac{|v_{n+j}- v_{n} |}{|\lambda_j-\lambda_0| |\lambda_{n+j+k}-\lambda_{n+j}| |\lambda_{n+k}-\lambda_{n}|}  \nonumber
 \\
 & &  \left. \qquad\qquad\qquad\qquad+ \frac{|v_{n+k}- v_{n+j+k} |}{|\lambda_j-\lambda_0| |\lambda_{n+j+k}-\lambda_{n+j}| |\lambda_{n+k}-\lambda_{n}|} \right).  \label{eq:useA.12timesA.2}
\end{eqnarray}
We observe that by Assumption~\ref{A.3new}:
\begin{eqnarray}
\frac{|\lambda_{n+k}- \lambda_n - \lambda_{n+j+k}+\lambda_{n+j}|}{|\lambda_j-\lambda_0| |\lambda_{n+j+k}-\lambda_{n+j}| |\lambda_{n+k}-\lambda_{n}|} &=&  
 \left| \frac{1}{\lambda_j-\lambda_0}\left(\frac{1}{\lambda_{n+j+k}-\lambda_{n+j} } - \frac{1}{\lambda_{n+k}-\lambda_n} \right) \right| \nonumber \\
&\leq&   c \left( 1+ \frac{1}{|\lambda_k-\lambda_0|}\right). \label{eq:first} 
\end{eqnarray}
Moreover, using the definition of $\| v\|_{T}$ and Assumption~\ref{A.1new}:
\begin{eqnarray}
 \frac{|v_{n+j}- v_{n} |}{|\lambda_j-\lambda_0| |\lambda_{n+j+k}-\lambda_{n+j}| |\lambda_{n+k}-\lambda_{n}|} 
&\leq&   \frac{ \| v\|_{T}}{ |\lambda_{n+j+k}-\lambda_{n+j}| |\lambda_{n+k}-\lambda_{n}|}  \nonumber \\
&\leq&   \| v\|_{T}c^2|k|^{2\gamma}, \label{eq:second}
\end{eqnarray}
and
\begin{eqnarray}
 \frac{|v_{n+k}- v_{n+j+k} |}{|\lambda_j-\lambda_0| |\lambda_{n+j+k}-\lambda_{n+j}| |\lambda_{n+k}-\lambda_{n}|} 
&\leq&   \frac{ \| \Theta_k v\|_{T}}{ |\lambda_{n+j+k}-\lambda_{n+j}| |\lambda_{n+k}-\lambda_{n}|}  \nonumber\label{eq:usetraslinv}  \\
&\leq&   \| v\|_{T}c^2|k|^{2\gamma}, \label{eq:third}
\end{eqnarray}
where in~(\ref{eq:third}) we used that $\|v\|_{T}=\|\Theta_k v\|_{T}$.
By plugging~\eqref{eq:first},~\eqref{eq:second},and ~\eqref{eq:third} in~\eqref{eq:useA.12timesA.2}, we have that
\begin{eqnarray}
&& \left| \frac{1}{\lambda_j - \lambda_0} \left(\frac{1}{ \lambda_{n+j+k}+ v_{n+j+k}- \lambda_{n+j}- v_{n+j} } -\frac{1}{ \lambda_{n+k}+ v_{n+k}- \lambda_n- v_n }\right)\right| \nonumber \\
& & \leq  \frac{1}{ \left( 1- 2c \| v \|_{T}\right)^2 }  \left( c\left(1+\frac{1}{|\lambda_k-\lambda_0|}\right)+ 2 \| v \|_{T} c^2 |k|^{2 \gamma} \right)  \nonumber\\
& & \leq 4 (c+c^2 |k|^{\gamma}+ 2 \| v \|_{T}c^2|k|^{2 \gamma})  \label{eq:useA.1d} \\
& & \leq  10c^2 |k|^{2 \gamma}, \label{eq:use14c1}
\end{eqnarray}
where  in~(\ref{eq:useA.1d})  we used  Assumption~\ref{A.1new} and in~(\ref{eq:use14c1}) we used  $\| v\|_{T} \leq 1/4c$ and $c\geq 1$.
Therefore,   we can conclude that
$$
\left\|\left(\Theta_k \tilde{\lambda}- \tilde{\lambda}\right)^{-1}\right\|_{T}  \leq 2 c |k|^{\gamma}+10 c^2 |k|^{2 \gamma} \leq 12 c^2 |k|^{2\gamma}.
$$
Now we prove $(ii)$. We define the sequence $w=(w_n)_{n \in \Z^d}$, with $w_n= \sum_{j=0}^{\ell} \e_j V^{(j)}_{n,n}$. Notice that, by Lemma~\ref{lemma:ineqnorm} $(iii)$, and by assumption~\eqref{eqn:hypsumeps},
$$
\|w\|_{T}=\left\| \sum_{j=0}^{\ell} \e_j \left[ V^{(j)} \right]\right\|_{\infty} \leq  \sum_{j=0}^{\ell}\left\| \e_j  V^{(j)} \right\|_{\alpha_j} \leq \frac{1}{4c}.
$$
Therefore, by $(i)$ we have that 
$$
\sup_{n \in \Z^d}  \frac{1}{\left| \lambda^{(\ell+1)}_{n+k}- \lambda_n^{(\ell+1) }\right|}\leq \left\|\left(\Theta_k \lambda^{(\ell+1)}- \lambda^{(\ell+1)}\right)^{-1}\right\|_{T}  \leq  12 c^2 |k|^{2\gamma}.
$$
Moreover, it is easy to check that the desired solution of~(\ref{eqn:homologicalth}) is $W^{(\ell)}$ such that for all $m,k \in \Z^d, k \neq 0$:
$$
W^{(\ell)}_{m,m+k}= \frac{\e_\ell V^{(\ell)}_{m,m+k}}{\lambda^{(\ell+1)}_{m+k}-\lambda^{(\ell+1)}_m}, \quad W^{(\ell)}_{m,m}=1.
$$
Notice that $W^{(\ell)}$ can be rewritten in terms of its main diagonals as $W_{0}^{(\ell)}={\bf 1}$, and  for all $ k \in \Z^d_0$
$$
W_k^{(\ell)}=\left(\Theta_k \lambda^{(\ell+1)}-\lambda^{(\ell+1)}\right)^{-1}\, \e_\ell V_k^{(\ell)}.
$$
Moreover, since $\left\|W^{(\ell)}_0-{\bf 1}\right \|_{T}=0$, we have that
\begin{eqnarray}
\left\|W^{(\ell)}-I\right\|_{\alpha_\ell - \delta}&=& \max\left\{\left\|W^{(\ell)}_0-{\bf 1}\right \|_{T}, \;\sup_{k \in \Z_0^d} \left\|W^{(\ell)}_k\right\|_{T} e^{(\alpha_\ell-\delta) |k|}\right\} \nonumber \\
&=& \sup_{k \in \Z_0^d} \left\|W^{(\ell)}_k\right\|_{T}  e^{(\alpha_\ell-\delta) |k|}.
\end{eqnarray}
Now, by~(\ref{eqn:estDiofantea}) we have that for all $k \in \Z^d_0$:
\begin{eqnarray*}
\left\|W^{(\ell)}_k\right\|_{T}&=&\left\|\left(\Theta_k \lambda^{(\ell+1)}- \lambda^{(\ell+1)}\right)^{-1}\, \e_\ell V^{(\ell)}_k \right\|_{T} \\
&\leq&\left \|\left(\Theta_k \lambda^{(\ell+1)}- \lambda^{(\ell+1)}\right)^{-1} \right\|_{T} \left\| \e_\ell V^{(\ell)}_k \right\|_{T} \\
& \leq & 12 c^2 |k|^{2 \gamma}  \left\| \e_\ell  V_k^{(\ell)} \right\|_{T}  \\
&\leq& 12 c^2 |k|^{2 \gamma} e^{-\alpha_\ell|k|} \| \e_\ell V^{(\ell)} \|_{\alpha_\ell},
\end{eqnarray*}
whence
\begin{eqnarray*}
\left\|W^{(\ell)}-I\right\|_{\alpha_\ell - \delta}&=& \sup_{k \in \Z_0^d} \left\|W_k^{(\ell)}\right \|_{T} e^{(\alpha_\ell-\delta)|k|} \\
&\leq& 12c^2  \left\| \e_\ell V^{(\ell)} \right\|_{\alpha_\ell} \sup_{k \in \Z^d}  |k|^{2 \gamma} e^{-\delta|k|} \\
&\leq& 12 c^2 \left(\frac{2 \gamma}{e \delta}\right)^{2 \gamma}\left\| \e_\ell V^{(\ell)}\right\|_{\alpha_\ell},
\end{eqnarray*}
where we have made use of the straightforward estimation
$$
\sup_{k \in \Z^d}  |k|^{2 \gamma} e^{-\delta|k|} \leq \max_{r>0} r^{2\gamma} e^{-\delta r} = \left(\frac{2 \gamma}{e \delta}\right)^{2 \gamma}.
$$

\end{proof}

\section{Proof of Theorem~\ref{thm:mainth2}} 
\label{sect:proof}
The proof will follow the steps outlined in Section~\ref{sect:roadmap}.
\subsection{Part $(1)$. KAM scheme}
In this part of the proof we formulate the KAM scheme and prove that it is possible to iterate the procedure. At each step $\ell \in \N$ of the scheme we want to find two operators $U^{(\ell+1)}, V^{(\ell+1)} \in {\mathcal L}\left(\ell^2\left(\Z^d\right)\right)$ such that
\begin{equation}\label{eqn:Ul1TUl}
{U^{(\ell+1)}}^{-1}(T+ \e V)U^{(\ell+1)}=T^{(\ell+1)}+ \e_{\ell+1} V^{(\ell+1)},
\end{equation}
where  $T^{(0)}=T$, $V^{(0)}=V$, $U^{(0)}=I$, 
$$
T^{(\ell+1)}=T+ \sum_{j=0}^{\ell} \e_j \left[V^{(j)}\right]
$$
and $\e_j:=\e^{2^j}$ for all $j=0, \dots, \ell+1$.

To do this we have to prove that  the homological equation
\begin{equation}
\label{hom01}
\left[T^{(\ell+1)}, W^{(\ell)}\right]+ \e_\ell V^{(\ell)}-\e_\ell \left[V^{(\ell)}\right]=0,
\end{equation}
admits an invertible solution ${W^{(\ell)}}$ and then define
$$
\e_{\ell+1} V^{(\ell+1)}:= {W^{(\ell)}}^{-1}\left(\e_\ell V^{(\ell)}-\e_\ell \left[V^{(\ell)}\right]\right)\left(W^{(\ell)}-I\right),
$$
and
$$
U^{(\ell+1)}:=U^{(\ell)}W^{(\ell)}.
$$

To this end we want to apply  Theorem~\ref{thm:propdiof}, so we need to prove that all the assumptions of Theorem~\ref{thm:propdiof} are satisfied.
We define  a  decreasing sequence $(\alpha_{p})_{p \in \N}$, with
$$
\alpha_p:=\alpha-2 \sum_{\nu=0}^{p-1} \sigma_{\nu},
$$
where 
$$
	\sigma_{\nu}=\frac{\sigma}{2^{\nu+2}}, \qquad  \sigma=\min\left\{1, \frac{\alpha}{2}\right\},
$$
 and we observe that 
$$
\lim_{p \to +\infty} \alpha_p= \alpha-\sigma >0.
$$
Then we consider the sequence $\left(\phi_p\right)_{p\in \N}$, with $\phi_0=1$ and
$$
\phi_p=  \prod_{\nu=0}^{p-1}\Phi(\sigma_{\nu})^{\frac{1}{2^{\nu+1}}}, \qquad p \geq 1
$$
where
$$
\Phi(x)=12 c^2 \left(\frac{2\gamma}{e }\right)^{2\gamma} x^{-4d-2 \gamma},
$$
and, by Lemma~\ref{prodinf}, we have 
\begin{equation}\label{eq:Phiphiinf}
\phi_{\infty}:=\lim_{p \to \infty} \phi_p=  12 c^2 \left( \frac{2 \gamma}{e}\right)^{2 \gamma} \left(\frac{4}{\sigma}\right)^{4d+2 \gamma}.
\end{equation}

We will prove that if $\e$ is small enough, namely if
\begin{equation}\label{eqn:hypeps}
|\e| \leq \e^*=\frac{\xi}{4c \|V\|_{\alpha}},
\end{equation}
with 
\begin{equation*}
	\xi= \frac{  \sigma^{4d+2\gamma}}{ 3c \left(\frac{2\gamma}{e }\right)^{2\gamma} 4^{4d+2\gamma}}\min\left\{ \frac{2^{4(2d+\gamma)} 3 c \left(\frac{2\gamma}{e }\right)^{2\gamma}  }{1+2^{4(2d+\gamma)} 3 c \left(\frac{2\gamma}{e }\right)^{2\gamma} }, 1-\frac{3^d}{2^{6d-1}}\right\}<1.
\end{equation*} 
then:
\begin{enumerate}
\item[(A)]  $\sum_{j=0}^{\ell} \left\| \e_j V^{(j)}\right\|_{\alpha_j} \leq \frac{1}{4c}$;
\item[(B)] $\left\| \e_{j}V^{(j)}\right\|_{\alpha_{j}} \leq \left(\frac{\xi}{4c} \phi_{j} \right)^{2^{j}}$, for all $j=0, \dots, \ell+1$;
\item[(C)] $\left\|U^{(j+1)}-U^{(j)}\right\|_{\alpha_{j+1}+\sigma_{j+1}} \leq\left( \frac{\xi }{4c} \phi_\infty \right)^{2^{j}}$,  for all $j=0, \dots, \ell$;
\item[(D)]  $\left\|{U^{(j+1)}}^{-1}-{U^{(j)}}^{-1}\right\|_{\alpha_{j+1}} \leq\left( \frac{\xi }{4c} \phi_\infty \right)^{2^{j}}$,  for all $j=0, \dots, \ell$.
\end{enumerate}

The proof of these conditions is done by induction on $\ell$.
\subsubsection{Base case, $\ell=0$}

For $\ell=0$ we have $\varepsilon_0=\varepsilon$, $\alpha_0=\alpha$, $\sigma_0=\sigma/4$ and $\phi_0=1$.
We want to solve the equation
\begin{equation}\label{eqn:homeq0}
\left[T^{(1)}, W^{(0)}\right]+\e_0 V^{(0)}-\e_0 \left[V^{(0)}\right]=0,
\end{equation}
under the condition $[W^{(0)}-I]=0$. $T^{(1)}$ is a diagonal operator 
with eigenvalues 
$$
\lambda^{(1)}_n= \lambda_n+\e_0 V^{(0)}_{n,n}, \quad n \in \Z^d.
$$ 
By~(\ref{eqn:hypeps})  we have immediately that condition $(A)$ holds
$$
 \left\| \e_0 V^{(0)}\right\|_{\alpha_0} = \left\| \e V \right\|_{\alpha} \leq \frac{\xi}{4c}<  \frac{1}{4c}.
$$

Then, by Theorem~\ref{thm:propdiof} the spectrum of $T^{(1)}$ is purely point with simple eigenvalues such that
$$
\sup_{n \in \Z^d}\frac{1}{\left|\lambda^{(1)}_{n+k}-\lambda^{(1)}_n \right|}\leq 12 c^2 |k|^{2 \gamma}, \qquad \forall k \in \Z^d_0.
$$
and the homological equation~(\ref{eqn:homeq0}) admits a unique solution $W^{(0)} \in M^{\alpha_0-\sigma_0}$ and
\begin{eqnarray}
\left\| W^{(0)}-I \right\|_{\alpha_0-\sigma_0} &\leq & 12 c^2 \left( \frac{2 \gamma}{e \sigma_0}\right)^{2 \gamma}  \|\e_0V\|_{\alpha_0} \nonumber \\
&= & \sigma_0^{4d} \Phi(\sigma_0) \|\e_0V\|_{\alpha_0} \nonumber \\
& < & \sigma_0^{4d} \Phi(\sigma_0) \frac{\xi}{4c}  \label{eqn:W0Phi}\\
&\leq & \sigma_0^{4d} \phi_{\infty} \frac{\xi}{4c} \label{eqn:Phimineq} \\
& < & \sigma_0^{4d} \label{eqn:1disphiinf}  \\
& < & \frac{1}{2}\left( \frac{\sigma_0}{3}\right)^{d} \label{sigma0} \\
&<& \frac{1}{4} \nonumber
\end{eqnarray}
where ~(\ref{eqn:Phimineq}) and~(\ref{eqn:1disphiinf}) follow respectively from~(\ref{philphiinf}) and~(\ref{eqn:sistema}) of Lemma~\ref{prodinf}, while~(\ref{sigma0}) follows from the definition of $\sigma_0$.
Hence, by Lemma~\ref{lemma:estinv},  $W^{(0)}$ is invertible and we get
\begin{eqnarray}
\left\|{W^{(0)}}^{-1}-I\right\|_{\alpha_0-2\sigma_0} &\leq & \frac{\left\|W^{(0)}-I\right\|_{\alpha_0-\sigma_0}}{1- \left(\frac{3}{\sigma_0}\right)^d \left\|W^{(0)}-I\right\|_{\alpha_0-\sigma_0}} \nonumber \\
& \leq & 2 \left\|W^{(0)}-I\right\|_{\alpha_0-\sigma_0} \nonumber \\
& \leq & 2 \sigma_0^{4d} \Phi(\sigma_0) \frac{\xi}{4c} \nonumber \\
& \leq & 2 \sigma_0^{4d}   \phi_{\infty} \frac{\xi }{4c} \label{eqn:estW0-1} \\
& < & \frac{1}{2}, \nonumber
\end{eqnarray}
and this yields $\left\|{W^{(0)}}^{-1}\right\|_{\alpha_0-2\sigma_0} \leq 3/2$. Then, we define
$$
\e_1 V^{(1)}:= {W^{(0)}}^{-1}\left(\e_0 V^{(0)}-\e_0 [V^{(0)}]\right)\left(W^{(0)}-I\right)
$$
$$
U^{(1)}:=W^{(0)},
$$
whence
$$
{U^{(1)}}^{-1}(T^{(0)}+\e_0 V^{(0)})U^{(1)}=T^{(1)}+\e_1 V^{(1)}.
$$

Now we prove condition $(B)$. First of all we have that by~(\ref{eqn:hypeps}) 
$$
\left \|\e_0 V^{(0)} \right\|_{\alpha_0}=\left \|\e V \right\|_{\alpha} \leq \frac{\xi}{4c}=\frac{\xi \phi_0}{4c},
$$
because $\phi_0=1$. Then we estimate $ \left\|\e_1V^{(1)}\right\|_{\alpha_1} $ by using  Lemma~\ref{lemma:prodotto} twice:
\begin{eqnarray}
 \left\|\e_1V^{(1)}\right\|_{\alpha_1} &=&  \left\|{W^{(0)}}^{-1}\left(\e_0 V^{(0)}-\e_0 [V^{(0)}]\right)\left(W^{(0)}-I\right)\right\|_{\alpha_0-2\sigma_0} \nonumber \\
& \leq &  \left( \frac{3}{\sigma_0} \right)^{d}\left\| {W^{(0)}}^{-1}\right\|_{\alpha_0-2\sigma_0} \left\| \left(\e_0 V^{(0)}-\e_0 [V^{(0)}]\right)\left(W^{(0)}-I\right)\right\|_{\alpha_0-\sigma_0}  \nonumber \\
& \leq &  \left( \frac{3}{\sigma_0} \right)^{2d} \left\| {W^{(0)}}^{-1}\right\|_{\alpha_0-2\sigma_0} \left\| \e_0 V^{(0)}-\e_0 [V^{(0)}]\right\|_{\alpha_0} \left\| W^{(0)}-I\right\|_{\alpha_0-\sigma_0} \nonumber \\
& \leq &  \left( \frac{3}{\sigma_0} \right)^{2d} \left\| {W^{(0)}}^{-1}\right\|_{\alpha_0-2\sigma_0}  \left\| \e_0 V^{(0)}\right\|_{\alpha_0} \left\| W^{(0)}-I\right\|_{\alpha_0-\sigma_0} \label{eq:menodiag} \\
& \leq & \left( \frac{3}{\sigma_0} \right)^{2d} \left(\frac{3}{2}\right)  \left(\frac{ \xi  }{4c}\right) \left(\sigma_0^{4d} \Phi(\sigma_0) \frac{\xi }{4c}\right) \label{eqn:estV1} \\
& = & \frac{3}{2} (3\sigma_0)^{2d} \Phi(\sigma_0) \left( \frac{\xi }{4c} \right)^2 \nonumber \\
&=& \frac{3}{2} (3  \sigma_0)^{2d} \left( \frac{\xi }{4c} \phi_1 \right)^2 \nonumber \\
&=& \frac{3}{2}\left(\frac{3}{4} \sigma\right)^{2d} \left( \frac{\xi }{4c} \phi_1 \right)^2 \nonumber \\
&\leq & \left( \frac{\xi }{4c} \phi_1 \right)^2,
\end{eqnarray}
where in~\eqref{eq:menodiag} we used Lemma~\ref{lemma:ineqnorm} (iii), and in~(\ref{eqn:estV1}) we used~(\ref{eqn:W0Phi}). 

Finally, we prove conditions $(C), (D)$. First observe that $\alpha_1 +\sigma_1 = \alpha_0 -2 \sigma_0 +\sigma_1 \leq \alpha_0-\sigma_0$. Therefore, by Lemma~\ref{lemma:ineqnorm} (i),
\begin{eqnarray}
\left\|U^{(1)}-U^{(0)}\right\|_{\alpha_1+\sigma_1} & \leq & \left\|U^{(1)}-U^{(0)}\right\|_{\alpha_0-\sigma_0} \nonumber \\
& = &  \left\|W^{(0)}-I \right\|_{\alpha_0-\sigma_0} \nonumber \\
& \leq & \sigma_0^{4d} \Phi(\sigma_0) \frac{\xi }{4c} \label{eqn:estW01} \\
& \leq &  \sigma_0^{4d}  \frac{\xi }{4c} \phi_{\infty} \nonumber \\
&= & \left(\frac{\sigma}{2^2}\right)^{4d} \frac{\xi }{4c} \phi_{\infty} \nonumber \\
& < & \frac{\xi }{4c} \phi_{\infty}, \nonumber
\end{eqnarray}
where in~(\ref{eqn:estW01}) we used~(\ref{eqn:W0Phi}).
Analogously, we have that
\begin{eqnarray}
\left\|{U^{(1)}}^{-1}-U^{(0)-1}\right\|_{\alpha_1} & \leq & \left\|{W^{(0)}}^{-1}-I\right\|_{\alpha_0-2\sigma_0} \nonumber \\
& \leq &  2 \sigma_0^{4d}   \phi_{\infty} \frac{\xi }{4c} \label{eqn:estinv} \\
&= & 2 \left(\frac{\sigma}{2^2}\right)^{4d} \frac{\xi }{4c} \phi_{\infty} \nonumber \\
& < & \frac{\xi }{4c} \phi_{\infty}. \nonumber 
\end{eqnarray}
where in~(\ref{eqn:estinv}) we used~(\ref{eqn:estW0-1}). 

\subsubsection{Induction step, $\ell\geq1$}
Now we prove the induction step, i.e. we assume that conditions $(A), (B), (C), (D)$ hold for some $\ell -1$, with $\ell\geq 1$, and we prove that they hold also for $\ell$. 
We want to solve the equation
\begin{equation}\label{eqn:homeqell}
\left[T^{(\ell+1)}, W^{(\ell)}\right]+\e_\ell V^{(\ell)}-\e_\ell \left[V^{(\ell)}\right]=0,
\end{equation}
under the condition $\left[W^{(\ell)}-I\right]=0$. $T^{(\ell +1)}= T+ \sum_{j=0}^{\ell} \e_j \left[V^{(j)}\right]$ is a diagonal operator, 
with eigenvalues 
$$
\lambda^{(\ell+1)}_m=\lambda_m+ \sum_{j=0}^\ell  \e_j V^{(j)}_{m,m}, \quad m \in \Z^d.
$$

We prove condition $(A)$: by the validity of condition $(B)$ for $j=0,\dots,\ell$, we have that
\begin{eqnarray}
\sum_{j=0}^{\ell} \left\| \e_jV^{(j)} \right\|_{\alpha_j} & \leq & \sum_{j=0}^{\ell}  \left( \frac{\xi }{4c} \phi_j\right)^{2^j} \nonumber \\
& \leq &  \sum_{j=0}^{\ell} \left( \frac{\xi }{4c} \phi_\infty \right)^{2^j}\frac{1}{\Phi(\sigma_j)} \label{eqn:phijphinf} \\
&=& \frac{\sigma^{4d+2\gamma}}{12 c^2 \left( \frac{2 \gamma}{e}\right)^{2 \gamma} 2^{2(4d+2 \gamma)}} \sum_{j=0}^{\ell} \left( \frac{\xi }{4c} \phi_\infty \right)^{2^j}\frac{1}{2^{j(4d+2 \gamma)}} \nonumber \\
&\leq& \frac{\sigma^{4d+2 \gamma}}{12 c^2 \left( \frac{2\gamma}{e}\right)^{2 \gamma} 2^{2(4d+2  \gamma)}} \sum_{k=1}^{\infty} \left( \frac{\xi }{4c} \phi_\infty \right)^{k} \nonumber \\
& \leq & \frac{1}{ 12 c^2 \left( \frac{2\gamma}{e}\right)^{2\gamma}2^{2(4d+2\gamma)}} \frac{\frac{\xi }{4c} \phi_\infty }{1-\frac{\xi }{4c} \phi_\infty }\nonumber  \\
& \leq & \frac{1}{4c}, \label{coeffmin1}
\end{eqnarray}
where in~(\ref{eqn:phijphinf}) and~(\ref{coeffmin1})  we used respectively~(\ref{philphiinf}) and~(\ref{eqn:sistema}) of Lemma~\ref{prodinf}. 

Therefore, by Theorem~\ref{thm:propdiof} (ii), $T^{(\ell+1)}$ has purely point spectrum with simple eigenvalues $\lambda^{(\ell+1)}=\left( \lambda^{(\ell+1)}_n \right)_{n \in \Z^d}$ such that
$$
\left\|\left(\Theta_k \lambda^{(\ell+1)}-\lambda^{(\ell+1)} \right)^{-1}\right\|_{T} \leq 12c^2 |k|^{2\gamma} \qquad \forall k \in \Z^d_0,
$$
and the homological equation~(\ref{eqn:homeqell}) admits a unique solution $W^{(\ell)} \in M^{\alpha_\ell-\sigma_\ell}$, with
\begin{eqnarray}
\left\| W^{(\ell)}-I \right\|_{\alpha_\ell-\sigma_\ell} &<& 12 c^2 \left(\frac{2 \gamma }{e \sigma_\ell}\right)^{2 \gamma }  \left\|\e_\ell V^{(\ell)}\right \|_{\alpha_\ell} \nonumber \\
& = & \sigma_\ell^{4d} \Phi(\sigma_\ell) \left\| \e_\ell V^{(\ell)}\right\|_{\alpha_\ell} \nonumber\\
&\leq & \sigma_\ell^{4d} \Phi(\sigma_\ell) \left(\frac{\xi }{4c} \phi_\ell \right)^{2^\ell} \label{estV}\\
&\leq & \sigma_\ell^{4d} \left(\frac{\xi }{4c} \phi_\infty \right)^{2^\ell} \label{estPhiphi}\\
&\leq & \frac{1}{2} \left( \frac{\sigma_\ell}{3}\right)^d.   \label{estPhiphi1}\\
& < & \frac{1}{4} \label{estsigmaell}
\end{eqnarray}
where in~(\ref{estV}) we used  again the validity condition $(B)$ for $j=\ell$, in~(\ref{estPhiphi}) and~(\ref{estPhiphi1}) we used respectively~(\ref{philphiinf}) and~(\ref{eqn:sistema}) of Lemma~\ref{prodinf}, and in~(\ref{estsigmaell}) we used the definition of $\sigma_\ell$. 

Therefore, by Lemma~\ref{lemma:estinv}, $W^{(\ell)}$ is invertible and
\begin{eqnarray}
\left\|{W^{(\ell)}}^{-1}-I\right\|_{\alpha_\ell-2\sigma_\ell} &\leq & \frac{\left\|W^{(\ell)}-I\right\|_{\alpha_\ell-\sigma_\ell}}{1- \left(\frac{3}{\sigma_\ell}\right)^d \left\|W^{(\ell)}-I\right\|_{\alpha_\ell-\sigma_\ell}} \nonumber \\
& \leq & 2 \left\|W^{(\ell)}-I\right\|_{\alpha_\ell-\sigma_\ell} \nonumber \\
& \leq & 2 \sigma_\ell^{4d} \Phi(\sigma_\ell) \left(\frac{\xi }{4c} \phi_\ell \right)^{2^\ell} \label{estPhiphi2} \\
& \leq & \frac{1}{2}, \nonumber
\end{eqnarray}
where in~(\ref{estPhiphi2}) we used~(\ref{estV}).
Thus, $\left\|{W^{(\ell)}}^{-1}\right\|_{\alpha_\ell-2\sigma_\ell} \leq 3/2$. 

Now, we define
$$
\e_{\ell+1} V^{(\ell+1)}:= {W^{(\ell)}}^{-1}\left(\e_\ell V^{(\ell)}-\e_\ell \left[V^{(\ell)}\right]\right)\left(W^{(\ell)}-I\right),
$$
$$
U^{(\ell+1)}:=U^{(\ell)}W^{(\ell)}, \qquad {U^{(\ell+1)}}^{-1}:={W^{(\ell)}}^{-1}{U^{(\ell)}}^{-1},
$$
and we have that
$$
{U^{(\ell+1)}}^{-1}\left(T^{(0)}+\e_0 V^{(0)}\right)U^{(\ell+1)}=T^{(\ell+1)}+\e_{\ell+1} V^{(\ell+1)}.
$$

In order to prove condition $(B)$, the only quantity that  have to be estimated is $\left\|\e_{\ell+1}V^{(\ell+1)}\right\|_{\alpha_{\ell+1}}$. We use Lemma~\ref{lemma:prodotto} twice:
\begin{eqnarray}
\left\|\e_{\ell+1}V^{(\ell+1)}\right\|_{\alpha_{\ell+1}} &=&  \left\|{W^{(\ell)}}^{-1}\left(\e_\ell V^{(\ell)}-\e_\ell \left[V^{(\ell)}\right]\right)\left(W^{(\ell)}-I\right)\right\|_{\alpha_\ell-2\sigma_\ell} \nonumber \\
& \leq & \left( \frac{2}{\sigma_\ell} \right)^d \left\| {W^{(\ell)}}^{-1}(\e_\ell V^{(\ell)}-\e_\ell [V^{(\ell)}])\right\|_{\alpha_\ell-2\sigma_\ell}  \left\| W^{(\ell)}-I\right\|_{\alpha_\ell-\sigma_\ell} \nonumber \\
& \leq &  \left( \frac{2}{\sigma_\ell} \right)^{2d} \| {W^{(\ell)}}^{-1}\|_{\alpha_\ell-2\sigma_\ell} \| \e_\ell V^{(\ell)}-\e_\ell [V^{(\ell)}]\|_{\alpha_\ell-\sigma_\ell} \| W^{(\ell)}-I\|_{\alpha_\ell-\sigma_\ell} \nonumber \\
& \leq &  \left( \frac{2}{\sigma_\ell} \right)^{2d} \left\| {W^{(\ell)}}^{-1}\right\|_{\alpha_\ell-2\sigma_\ell} \left\| \e_\ell V^{(\ell)}-\e_\ell [V^{(\ell)}]\right\|_{\alpha_\ell} \left\| W^{(\ell)}-I\right\|_{\alpha_\ell-\sigma_\ell} \nonumber \\
& \leq &  \left( \frac{2}{\sigma_\ell} \right)^{2d} \left\| {W^{(\ell)}}^{-1}\right\|_{\alpha_\ell-2\sigma_\ell} \left\| \e_\ell V^{(\ell)}\right\|_{\alpha_\ell} \left\| W^{(\ell)}-I\right\|_{\alpha_\ell-\sigma_\ell} \nonumber \\
& \leq & \left( \frac{2}{\sigma_\ell} \right)^{2d} \left(\frac{3}{2}\right) \left(\frac{ \xi  }{4c} \phi_\ell \right)^{2^\ell} \left(\sigma_\ell^{4d} \Phi(\sigma_\ell)\left(\frac{ \xi  }{4c} \phi_\ell \right)^{2^\ell}\right) \nonumber \\
& = & \frac{3}{2} 2^{2d}   \sigma_\ell^{2d} \Phi(\sigma_\ell) \left(\frac{ \xi  }{4c} \phi_\ell \right)^{2^{\ell+1}}\nonumber \\
&=& \frac{3}{2} 2^{2d}   \sigma_\ell^{2d} \left(\frac{ \xi }{4c} \phi_{\ell+1} \right)^{2^{\ell+1}} \nonumber \\
&=& \frac{3}{2}\left(\frac{\sigma}{2^{\ell+1}}\right)^{2d}\left(\frac{ \xi }{4c} \phi_{\ell+1} \right)^{2^{\ell+1}} \nonumber \\
&\leq &\left(\frac{ \xi  }{4c} \phi_{\ell+1} \right)^{2^{\ell+1}}. \nonumber
\end{eqnarray}

Finally, we prove conditions $(C), (D)$. First of all  observe that the induction step of condition $(C)$ implies that
\begin{eqnarray}
\left\|U^{(\ell)}-I\right\|_{\alpha_\ell+\sigma_\ell} & \leq & \sum_{j=0}^{\ell-1} \left\|U^{(j+1)}-U^{(j)}\right\|_{\alpha_\ell+\sigma_\ell} \nonumber \\
& \leq & \sum_{j=0}^{\ell-1} \left\|U^{(j+1)}-U^{(j)}\right\|_{\alpha_{j+1}+\sigma_{j+1}} \nonumber \\
& \leq & \sum_{j=0}^{\ell-1} \left( \frac{\xi }{ 4c} \phi_{\infty}\right)^{2^{j}} \nonumber \\
& \leq & \sum_{i=1}^{\infty} \left( \frac{\xi }{ 4c} \phi_{\infty}\right)^{i} \nonumber \\
& = & \frac{\frac{\xi }{ 4c} \phi_{\infty}}{1-\frac{\xi }{ 4c} \phi_{\infty}}, \label{geomser1}
\end{eqnarray}
where in~(\ref{geomser1}) we used~(\ref{eqn:sistema}). Therefore,
\begin{equation}\label{eqn:estUelle}
\left\|U^{(\ell)}\right\|_{\alpha_\ell}\leq \left\|U^{(\ell)}\right\|_{\alpha_\ell+\sigma_\ell} \leq 1+\frac{\frac{\xi }{4c} \phi_{\infty}}{1-\frac{\xi }{4c} \phi_{\infty}}= \frac{1}{1-\frac{\xi }{4c} \phi_{\infty}}.
\end{equation}
We observe that $\alpha_{\ell+1}+\sigma_{\ell+1} \leq \alpha_\ell-\sigma_\ell$, whence
\begin{eqnarray}
\left\| U^{(\ell+1)}-U^{(\ell)}\right\|_{\alpha_{\ell+1}+\sigma_{\ell+1}} & = & \left\| U^{(\ell)}\left(W^{(\ell)}-I\right) \right\|_{\alpha_{\ell+1}+\sigma_{\ell+1}} \nonumber \\
& \leq & \left\| U^{(\ell)}\left(W^{(\ell)}-I\right) \right\|_{\alpha_\ell-\sigma_\ell} \nonumber \\
&\leq& \left( \frac{3}{\sigma_\ell}\right)^d \left\|U^{(\ell)}\right\|_{\alpha_\ell} \left\| W^{(\ell)}-I\right\|_{\alpha_\ell-\sigma_\ell}\nonumber\\
& \leq & \left( \frac{3}{\sigma_\ell}\right)^d \frac{1}{1-\frac{\xi }{4c}\phi_{\infty}} \sigma_\ell^{4d} \left(\frac{\xi }{4c} \phi_\infty \right)^{2^\ell} \label{stimaUelle}  \\
& = & \frac{3^{d}}{2^{3d(\ell+2)}} \frac{\sigma^{3d}}{1-\frac{\xi }{4c}\phi_{\infty}} \left(\frac{\xi }{4c} \phi_\infty \right)^{2^\ell} \nonumber \\
& \leq & \frac{3^d}{2^{6d}\left(1-\frac{\xi }{4c}\phi_{\infty} \right)}\left(\frac{\xi}{4c} \phi_\infty \right)^{2^\ell} \nonumber \\
& \leq & \left(\frac{\xi }{4c} \phi_\infty \right)^{2^\ell} , \label{coeffminore11}
\end{eqnarray}
where in~(\ref{stimaUelle}) we used~(\ref{eqn:estUelle}) and~\eqref{estPhiphi}, while in~(\ref{coeffminore11}) we used~(\ref{eqn:sistema}) of Lemma~\ref{prodinf}. 

Finally, we want to estimate $\left\| {U^{(\ell+1)}}^{-1}-{U^{(\ell)}}^{-1}\right\|_{\alpha_{\ell+1}}$. The induction step of condition $(D)$ implies that, analogously to~\eqref{geomser1},
$$
\left\|{U^{(\ell)}}^{-1}-I\right\|_{\alpha_\ell} \leq \frac{\frac{\xi}{ 4c} \phi_{\infty}}{1-\frac{\xi}{4c} \phi_{\infty}},
$$
whence
\begin{equation}\label{eqn:estUelle-1}
\left\|{U^{(\ell)}}^{-1}\right\|_{\alpha_\ell} \leq 1+\frac{\frac{\xi }{4c} \phi_{\infty}}{1-\frac{\xi }{4c} \phi_{\infty}}= \frac{1}{1-\frac{\xi }{4c} \phi_{\infty}}.
\end{equation}
Therefore,
\begin{eqnarray}
\left\| {U^{(\ell+1)}}^{-1}-{U^{(\ell)}}^{-1}\right\|_{\alpha_{\ell+1}} & = & \left\| \left({W^{(\ell)}}^{-1}-I\right) {U^{(\ell)}}^{-1}\right\|_{\alpha_{\ell}-2\sigma_\ell} \nonumber \\
&\leq& \left( \frac{3}{\sigma_\ell}\right)^d \left\| {W^{(\ell)}}^{-1}-I\right\|_{\alpha_\ell-2\sigma_\ell}  \left\|{U^{(\ell)}}^{-1}\right\|_{\alpha_\ell-\sigma_\ell} \nonumber \\
&\leq& \left( \frac{3}{\sigma_\ell}\right)^d  \left\| {W^{(\ell)}}^{-1}-I\right\|_{\alpha_\ell-2\sigma_\ell} \left\|{U^{(\ell)}}^{-1}\right\|_{\alpha_\ell}  \nonumber\\
& \leq & \left( \frac{3}{\sigma_\ell}\right)^d  2 \sigma_\ell^{4d} \left(\frac{\xi }{4c} \phi_\infty \right)^{2^\ell} \frac{1}{1-\frac{\xi }{4c}\phi_{\infty}} \label{stimaUelle-1} \\
& \leq & \frac{3^d}{2^{6d-1}\left(1-\frac{\xi }{4c}\phi_{\infty} \right)}\left(\frac{\xi }{4c} \phi_\infty \right)^{2^\ell} \nonumber \\
& \leq & \left(\frac{\xi}{4c} \phi_\infty \right)^{2^\ell}, \label{coeffminore12}
\end{eqnarray}
where in~(\ref{stimaUelle-1}) we used~(\ref{eqn:estUelle-1}) and~\eqref{estPhiphi2}, while in~(\ref{coeffminore12}) we used~(\ref{eqn:sistema}) of Lemma~\ref{prodinf}.

\subsection{Part $(2)$. Convergence of the KAM series} 
In this part of the proof we prove that the  KAM iterative scheme converges, i.e. that the sequence of transformations $\left( U^{(\ell)} \right)_{\ell \in \N}$ converges to an invertible operator $U \in {\mathcal L}\left(\ell^2(\Z^d)\right)$, and the sequence of operators $\left( T^{(\ell)} \right)_{\ell \in \N}$ strongly converges to an operator 
$$
T^{(\infty)}=T+ \sum_{j=0}^{+\infty} \e_j \left[V^{(j)}\right],
$$
so that, by passing to the limit $\ell \to + \infty$ in~(\ref{eqn:Ul1TUl}), one obtains
\begin{equation}\label{eqn:Uinf}
U^{-1}(T+\e V)U=T^{(\infty)}.
\end{equation}
We first prove that the operator sequences $\left(U^{(\ell)}\right)_{\ell \in \N}$ and $\left({U^{(\ell)}}^{-1}\right)_{\ell \in \N}$ are  Cauchy in the Banach space $M^{\alpha-\sigma}$. For all $p \in \N$ we have that
\begin{eqnarray}
\left\|U^{(\ell+p)}-U^{(\ell)}\right\|_{\alpha-\sigma}& \leq & \sum_{j=0}^{p-1} \left\|U^{(\ell+j+1)}-U^{(\ell+j)}\right\|_{\alpha-\sigma} \nonumber \\
& \leq & \sum_{j=0}^{p-1} \left\|U^{(\ell+j+1)}-U^{(\ell+j)}\right\|_{\alpha_{\ell+j+1}+\sigma_{\ell+j+1}} \nonumber \\
& \leq & \sum_{j=0}^{p-1} \left( \frac{\xi }{4c} \phi_\infty \right)^{2^{\ell+j}} \nonumber \\
& \leq &\frac{ \left( \frac{\xi }{4c} \phi_\infty \right)^{2^{\ell}} }{1- \left( \frac{\xi }{4c} \phi_\infty \right)^{2^{\ell}} } \to 0, \quad \text{as }\ell \to \infty. 
\end{eqnarray}
In the same way,  for all $p \in \N$:
\begin{eqnarray}
\left\|{U^{(\ell+p)}}^{-1}-{U^{(\ell)}}^{-1}\right\|_{\alpha-\sigma}& \leq & \sum_{j=0}^{p-1} \left\|{U^{(\ell+j+1)}}^{-1}-{U^{(\ell+j)}}^{-1}\right\|_{\alpha-\sigma} \nonumber \\
& \leq & \sum_{j=0}^{p-1} \left\|{U^{(\ell+j+1)}}^{-1}-{U^{(\ell+j)}}^{-1}\right\|_{\alpha_{\ell+j+1}} \nonumber \\
& \leq & \sum_{j=0}^{p-1} \left( \frac{\xi }{4c} \phi_\infty \right)^{2^{\ell+j}} \nonumber \\
& \leq &\frac{ \left( \frac{\xi }{4c} \phi_\infty \right)^{2^{\ell}} }{1- \left( \frac{\xi }{4c} \phi_\infty \right)^{2^{\ell}} } \to 0, \quad \ell \to \infty.  
\end{eqnarray}
Since $M^{\alpha-\sigma}$ is a Banach space, there exists $U \in M^{\alpha-\sigma}$ invertible such that
$$
\lim_{\ell \to \infty} \left\|U^{(\ell)}-U\right\|_{\alpha -\sigma}=\lim_{\ell \to \infty} \left\|{U^{(\ell)}}^{-1}-U^{-1}\right\|_{\alpha -\sigma}=0.
$$
Moreover, the series
\begin{equation}\label{eqn:estserie}
 \sum_{j=0}^{+\infty} \e_j \left[V^{(j)} \right]
\end{equation}
is convergent in $M^{\infty}$. Indeed,
\begin{eqnarray}
\left\| \sum_{j=0}^{+\infty} \e_j \left[V^{(j)} \right] \right\|_{\infty}   & \leq &   \sum_{j=0}^{+\infty} \left\|\e_j V^{(j)}\right\|_{\alpha_j}  \nonumber \\
& \leq & \sum_{j=0}^{+\infty}  \left( \frac{\xi }{4c} \phi_j\right)^{2^j} \nonumber \\
& \leq &  \sum_{j=0}^{+\infty} \left( \frac{\xi }{4c} \phi_\infty \right)^{2^j}\frac{1}{\Phi(\sigma_j)} \label{eqn:phijphinf1} \\
&=&  \frac{\sigma^{4d+2\gamma}}{12 c^2 \left( \frac{2 \gamma}{e}\right)^{2 \gamma} 2^{2(4d+2 \gamma)}} \sum_{j=0}^{+\infty} \left( \frac{\xi }{4c} \phi_\infty \right)^{2^j}\frac{1}{2^{j(4d+2 \gamma)}} \nonumber \\
&\leq& \frac{\sigma^{4d+2 \gamma}}{12 c^2 \left( \frac{2\gamma}{e}\right)^{2 \gamma} 2^{2(4d+2  \gamma)}} \sum_{k=1}^{\infty} \left( \frac{\xi }{4c} \phi_\infty \right)^{k} \nonumber \\
& \leq & \frac{1}{ 12 c^2 \left( \frac{2\gamma}{e}\right)^{2\gamma} 2^{2(4d+2\gamma)}} \frac{\frac{\xi }{4c} \phi_\infty }{1-\frac{\xi }{4c} \phi_\infty } \label{coeffmin12} \label{eqn:finalestnorminf}\nonumber \\
&\leq & \frac{1}{4c}     \label{estnormintorno}                                    
\end{eqnarray}
where in~(\ref{eqn:phijphinf1}) and~(\ref{estnormintorno})  we used respectively~(\ref{philphiinf}) and~(\ref{eqn:sistema}) of Lemma~\ref{prodinf}. 

Therefore, by Lemma~\ref{lemma:ineqnorm} (v), the series in~(\ref{eqn:estserie}) converges in norm. Since for all $\ell \in \N$: 
\begin{equation}\label{eqn:approxl}
{U^{(\ell)}}^{-1}(T+ \e V)U^{(\ell)}=T+ \sum_{j=0}^{\ell}\e_j \left[V^{(j)}\right],
\end{equation}
and since the convergence in $M^{\alpha-\sigma}$ and $M^{\infty}$  implies the strong convergence of operators, we have immediately, passing to the limit $\ell \to +\infty$ in~(\ref{eqn:approxl}), that
\begin{eqnarray}
\label{Tinfinito}
U^{-1}(T+ \e V)U=T^{(\infty)}.
\end{eqnarray}
Hence the operators $T+\e V$ and $T^{(\infty)}$ are isospectral.

Notice that by~(\ref{estnormintorno})
\begin{equation}
\left\|\sum_{j=0}^{\infty}\e_j \left[V^{(j)}\right]\right\|_{\infty} = \left\|\sum_{j=0}^{\infty}\e_j V^{(j)}_{0}\right\|_{T} \leq \frac{1}{4c}.                                                                  
\end{equation}
Thus, by  Theorem~\ref{thm:propdiof}, the operator  $T^{(\infty)}$ has a purely point spectrum with simple eigenvalues $\left( \lambda_n(\e) \right)_{n \in \Z^d}$, $\lambda_n(\e)=\lambda_n+\sum_{j=0}^{+\infty}\e_j V^{(j)}_{n,n}$, such that for all $k \in \Z^d$, $k \neq 0$:
\begin{equation}\label{eq:dioflimit}
\sup_{n \in \Z^d} \frac{1}{\left| \lambda_{n+k}(\e)-\lambda_n(\e)  \right|} \leq 12c^2 |k|^{2\gamma},
\end{equation}
that is~(\ref{Diofanto3}). 

\subsection{Part $(3)$. Unitarity}
Now we prove that if $V=V^*$ then for all $n,m \in \Z^d$:
$$
\langle e_n, U^* U e_m \rangle= \delta_{n,m} \|U e_n\|^2.  
$$
In fact, since $T+ \e V$ is self-adjoint its eigenvalues are real, hence for all $n,m \in \Z^d$:
\begin{eqnarray}
\lambda_n(\e) \langle U e_n, U e_m \rangle & = &  \langle \lambda_n(\e) U e_n, U e_m \rangle \nonumber \\
& = &  \langle (T+ \e V) U e_n, U e_m \rangle \nonumber \\
& = &  \langle  U e_n, (T+ \e V)^*U e_m \rangle \nonumber \\
& = &  \langle  U e_n, (T+ \e V)U e_m \rangle \nonumber \\
& = &  \langle  U e_n, \lambda_m(\e)U e_m \rangle \nonumber \\
& = &  \lambda_m(\e) \langle  U e_n, U e_m \rangle, \nonumber
\end{eqnarray}
hence $\left(\lambda_n(\e)-\lambda_m(\e) \right)\langle U e_n, U e_m \rangle=0$. But since $\lambda_n(\e)\neq \lambda_m(\e)$ if $n \neq m$, we obtain $\langle U e_n, U e_m \rangle=\langle e_n, U^* U e_m \rangle=\|U e_n\|^2 \delta_{n,m}$. Therefore $U^*U$ is diagonal 
 with eigenvalues $\|U e_n\|^2 >0$, $n \in \Z^d$.
The  operator 
$$
\tilde{U}:=U (U^{*}U)^{-1/2}
$$
is thus unitary, and  since $T^{(\infty)}$ commutes with $(U^*U)^{1/2} $ we get
\begin{eqnarray*}
\tilde{U}^{-1}(T+ \e V)\tilde{U}&=&(U^*U)^{1/2}U^{-1}(T+ \e V)U (U^*U)^{-1/2} \\                                                                &=&(U^*U)^{1/2} T^{(\infty)}(U^*U)^{-1/2} \\                                                                &=&T^{(\infty)}.
\end{eqnarray*}
Therefore, the operator $T + \e V$ is unitarily equivalent to $T^{(\infty)}$. Moreover, the orthonormal eigenvectors of $T+\e V$ are 
$$
u_n(\e):= \frac{U e_n}{\|U e_n\|}, \qquad n \in \Z^d.
$$

Moreover, since $U \in M^{\alpha- \sigma}$, we have that for all $n, j \in \Z^d$:
$$
|\langle e_j, u_n(\e)\rangle|=\frac{|U_{j,n}|}{\|U e_n\|} \leq \frac{\| U\|_{\alpha-\sigma}}{\|U e_n\|} e^{-(\alpha-\sigma)|j-n|},
$$
that is~(\ref{eqn:exploc}) with
\begin{equation}\label{eqn:Cn}
C_n=\frac{\| U\|_{\alpha-\sigma}}{\|U e_n\|}.
\end{equation}

This concludes the proof of Theorem~\ref{thm:mainth2}.

\section*{Acknowledgments}
We would like to express our deep gratitude to Sandro Graffi for valuable and constructive comments, enthusiastic encouragement, and useful discussions during the development of this research work.
This work was partially supported by Istituto Nazionale di Fisica Nucleare (INFN), Italy through the project ``QUANTUM'' and the Italian National Group of Mathematical Physics (GNFM-INdAM).

\appendix

\section{Technical Lemmas}\label{sect:lemmas}
In this appendix there are some technical lemmas useful to prove Theorem~\ref{thm:mainth2}. Lemma~\ref{lemma:prodotto},~\ref{lemma:estinv} are essentially the ones  proved in~\cite{ref:Poeschel} but here the results are adapted to our definitions and notations. We decided to insert the proof here for the sake of completeness. 
\begin{lemma}\label{lemma:prodotto}
Let $\alpha >0$ and let $0< \delta < \alpha$. If $X \in M^{\alpha-\delta}$ and $Y \in M^{\alpha}$, then $X Y \in M^{\alpha-\delta}$ and
\begin{equation} \label{eqn:estprod}
\| X Y \|_{\alpha-\delta} \leq q(\delta)\| X\|_{\alpha-\delta} \|Y\|_{\alpha},
\end{equation}
where 
$$
q(\delta)= \sum_{k \in \Z^d}   e^{-\delta |k|}=  \left( \frac{1+ e^{-\delta}}{1- e^{-\delta}} \right)^d. 
$$
Moreover if $\delta <1 $ then $q(\delta)<\left( \frac{3}{\delta}\right)^d$.
\end{lemma}
\begin{proof}
It suffices to prove the assertion when $\| X\|_{\alpha-\delta}=\| Y\|_{\alpha}=1$. We observe that for all $m,k \in \Z^d$:
\begin{equation}
(X Y)_{m,m+k}= \sum_{l \in \Z^d} X_{m,m+l} Y_{m+l,m+k}=  \sum_{l \in \Z^d} (X_{l})_m (\Theta_l Y_{k-l})_m,
\end{equation}
so  that 
$$
(XY)_k=\sum_{l \in \Z^d} X_{l} \,\Theta_l Y_{k-l}, \quad \forall k \in \Z^d.
$$
Thus,
\begin{eqnarray}
\| (X Y)_k \|_{T} &\leq & \sum_{l \in \Z^d}\| X_{l} \,\Theta_l Y_{k-l}\|_{T} \nonumber \\
& \leq & \sum_{l \in \Z^d}\| X_{l}\|_{T} \| \Theta_l Y_{k-l}\|_{T}  \nonumber \\
& = & \sum_{l \in \Z^d}\| X_{l}\|_{T} \| Y_{k-l}\|_{T} \label{eqn:contrnorm} \\
& \leq & \sum_{l \in \Z^d} e^{-(\alpha-\delta)|l|}e^{-\alpha|k-l|} \nonumber \\
&=& \sum_{l \in \Z^d} e^{-(\alpha-\delta)|k-l|}e^{-\alpha|l|} \nonumber \\
& \leq & \sum_{l \in \Z^d} e^{-(\alpha-\delta)(|k|-|l|)}e^{-\alpha|l|} \nonumber \\
& = & e^{-(\alpha-\delta)|k|}  \sum_{l \in \Z^d} e^{-\delta|l|} \nonumber \\
& = & e^{-(\alpha-\delta)|k|} q(\delta), \nonumber
\end{eqnarray}
where in~(\ref{eqn:contrnorm}) we used Lemma~\ref{lemma:contrnorm} (iii). Therefore,
$$
\| X Y \|_{\alpha-\delta} = \sup_{k \in \Z^d} \| (X Y)_k \|_{T}\, e^{(\alpha-\delta)|k|} \leq q(\delta).
$$
The explicit expression of $q(\alpha)$ is obtained by using~\eqref{eq:sumk}, and it is easily shown to be bounded by $(3/\delta)^d$ for $\delta<1$.
\end{proof}

\begin{lemma} \label{lemma:estinv}
Let $0 < \delta <\alpha $, $\delta \leq 1$ and $X \in M^{\alpha}$. If $\| X-I\|_{\alpha} < \left( \frac{\delta}{3}\right)^d$, then  $X$ is invertible in $M^{\alpha - \delta}$ and
$$
\left\| X^{-1}-I \right\|_{\alpha-\delta} < \frac{\| X-I\|_{\alpha} }{1-\left( \frac{3}{\delta}\right)^d\| X-I\|_{\alpha} }.
$$
\end{lemma}
\begin{proof}
We write the Neumann series of $X^{-1}$:
$$
X^{-1}=(I-(I-X))^{-1}= \sum_{l =0}^{+\infty}(I-X)^l,
$$
thus formally
$$
\left\|X^{-1}-I\right\|_{\alpha-\delta}\leq \sum_{l =1}^{+\infty}\|(I-X)^l\|_{\alpha-\delta}.
$$
By Lemma~\ref{lemma:prodotto}, we have 
$$
\left\|(I-X)^2 \right\|_{\alpha-\delta}=\|(I-X) (I-X)\|_{\alpha-\delta}\leq q(\delta)\| I-X \|_{\alpha-\delta}\|I-X\|_{\alpha}< \left( \frac{3}{\delta}\right)^{d}\|I-X\|_{\alpha}^2
$$
and in general $\left\|(I-X)^l \right\|_{\alpha-\delta} < \left( \frac{3}{\delta}\right)^{d(l-1)}\|I-X\|_{\alpha}^l$. Therefore it results that
$$
\left\|X^{-1}-I\right\|_{\alpha-\delta} \leq \sum_{l=1}^{+\infty}\left\|(I-X)^l \right\|_{\alpha-\delta} < \frac{\| I-X\|_{\alpha} }{1-\left( \frac{3}{\delta}\right)^d \| I-X\|_{\alpha} }.
$$
\end{proof}

\begin{lemma} \label{prodinf}
Let $d \in\N^*$, $c \geq 1$, $\gamma >0$, $ 0< \sigma \leq 1$, and $(\sigma_\nu)_{\nu \in \N}$ with $\sigma_\nu=\sigma/2^{\nu+2}$. Define $\phi_0=1$, and for all $\ell \geq 1$
$$
\phi_\ell=\prod_{\nu=0}^{\ell-1}\Phi(\sigma_{\nu})^{\frac{1}{2^{\nu+1}}},
$$
where
$$
\Phi(x)=12 c^2 \left(\frac{2 \gamma}{e }\right)^{2 \gamma} x^{-4d-2 \gamma}, \qquad x>0.
$$
Then
\begin{enumerate}
\item
\begin{equation} \label{infprodphi}
\phi_\infty:=\lim_{\ell \to \infty} \phi_\ell = 12 c^2 \left( \frac{2 \gamma}{e}\right)^{2 \gamma} \left(\frac{4}{\sigma}\right)^{4d+2 \gamma};
\end{equation}
\item  For all $\ell\in \N$
\begin{equation}\label{philphiinf}
\Phi(\sigma_\ell) \phi_\ell^{2^\ell} \leq \phi_\infty^{2^\ell};
\end{equation}
\item The largest solution $\xi$ of the following system
\begin{equation}\label{eqn:sistema}
\begin{cases}
\frac{\xi }{4c} \phi_\infty < 1 \\
\frac{1}{2^{2(4d+2 \gamma)} 12 c^2 \left( \frac{2 \gamma}{e}\right)^{2\gamma}} \frac{\frac{\xi }{4c} \phi_\infty }{1-\frac{\xi }{4c} \phi_\infty } \leq \frac{1}{4c}  \\
\frac{3^d}{2^{6d-1}\left(1-\frac{\xi }{4c}\phi_{\infty} \right)} \leq 1 
\end{cases}
\end{equation}
is
\begin{equation}\label{defn:xi}
\xi= \frac{  \sigma^{4d+2\gamma}}{ 3c \left(\frac{2\gamma}{e }\right)^{2\gamma} 4^{4d+2\gamma}}\min\left\{ \frac{2^{4(2d+\gamma)} 3 c \left(\frac{2\gamma}{e }\right)^{2\gamma}  }{1+2^{4(2d+\gamma)} 3 c \left(\frac{2\gamma}{e }\right)^{2\gamma} }, 1-\frac{3^d}{2^{6d-1}}\right\}<1.
\end{equation}
\end{enumerate}
\end{lemma}
\begin{proof}
First we prove $(\ref{infprodphi})$:
\begin{eqnarray}
\prod_{\nu=0}^{\infty}\Phi(\sigma_{\nu})^{\frac{1}{2^{\nu+1}}} & =& \prod_{\nu=0}^{\infty} \left( \frac{12 c^2 \left( \frac{2\gamma}{e}\right)^{2\gamma} 2^{2(4d+2\gamma)}}{\sigma^{4d+2\gamma}}\right)^{\frac{1}{2^{\nu+1}}} \prod_{\nu=0}^{\infty}  2^{\frac{(4d+2\gamma)\nu}{2^{\nu+1}}}  \label{eqn:sums} \\
& = &  \frac{12 c^2 \left( \frac{2\gamma}{e}\right)^{2\gamma} 2^{2(4d+2\gamma)}}{\sigma^{4d+2\gamma}}, \nonumber
\end{eqnarray}
where in~(\ref{eqn:sums}) we used that 
$$
\sum_{\nu =0}^{\infty} \frac{1}{2^{\nu + 1}}= 1, \qquad  \sum_{\nu =0}^{\infty} \frac{\nu}{2^{\nu + 1}}= 1.
$$

Now we prove $(\ref{philphiinf})$.  First we observe that $\Phi$ is a decreasing function and that $(\sigma_\nu)_{\nu \in \N}$ is a decreasing sequence, therefore $(\Phi(\sigma_\nu))_{\nu \in \N}$ is an increasing sequence.
Now we have
\begin{eqnarray}
\Phi(\sigma_\ell) \phi_\ell^{2^\ell} & = & \Phi(\sigma_\ell)^{\sum_{j=\ell}^\infty 2^{\ell-j-1}} \phi_\ell^{2^\ell} \nonumber \\
& = & \left( \Phi(\sigma_\ell)^{\sum_{j=\ell}^\infty 2^{-j-1}} \phi_\ell \right)^{2^\ell} \nonumber \\
& = & \left( \prod_{j=\ell}^{\infty}\Phi(\sigma_\ell)^{ \frac{1}{2^{j+1}}} \phi_\ell \right)^{2^\ell} \nonumber \\
& \leq & \left( \prod_{j=\ell}^{\infty}\Phi(\sigma_j)^{ \frac{1}{2^{j+1}}} \phi_\ell \right)^{2^\ell} \nonumber \\
& = & \left( \prod_{j=\ell}^{\infty}\Phi(\sigma_j)^{ \frac{1}{2^{j+1}}} \prod_{\nu=0}^{\ell-1} \Phi(\sigma_\nu)^{\frac{1}{2^{\nu+1}}} \right)^{2^\ell} \nonumber \\
& = & \left(\prod_{\nu=0}^{\infty} \Phi(\sigma_\nu)^{\frac{1}{2^{\nu+1}}} \right)^{2^\ell} \nonumber \\
&=& \phi_\infty^{2^\ell}.
\end{eqnarray}
Finally, the solution of the system~(\ref{eqn:sistema}) is
$$
\xi   \leq \frac{4c}{\phi_\infty} \min\left\{ \frac{2^{2(4d+2\gamma)}  3 c \left( \frac{2\gamma}{e}\right)^{2\gamma}}{1+2^{2(4d+2\gamma)} 3c \left( \frac{2\gamma}{e}\right)^{2\gamma}}, 1-\frac{3^d}{2^{6d-1}}\right\}.
$$
Using the value of $\phi_\infty$ we have that the largest solution is
$$
\xi= \frac{  \sigma^{4d+2\gamma}}{ 3c \left(\frac{2\gamma}{e }\right)^{2\gamma} 4^{4d+2\gamma}}\min\left\{ \frac{2^{4(2d+\gamma)} 3 c \left(\frac{2\gamma}{e }\right)^{2\gamma}  }{1+2^{4(2d+\gamma)} 3 c \left(\frac{2\gamma}{e }\right)^{2\gamma} }, 1-\frac{3^d}{2^{6d-1}}\right\}.
$$
Notice that, since $\sigma \leq 1$, $c \geq 1$ we have that 
$$
\xi < \frac{1}{ 3 \left(\frac{2\gamma}{e }\right)^{2\gamma} 4^{4d+2\gamma}}<1.
$$
\end{proof}

\end{document}